\newif\ifpdf
\ifx\pdfoutput\undefined
\pdffalse 
\else
\pdfoutput=1 
\pdftrue \fi

\newif\iffinal
\finaltrue 

\documentclass[reqno,twoside,11pt]{amsart}
\iffinal\else\usepackage[notref,notcite]{showkeys}\fi
\iffinal\else\IfFileExists{pdfsync.sty}{\usepackage{pdfsync}}{}\fi
\usepackage{cite}
\usepackage{amsmath}
\usepackage{amsfonts}
\usepackage{amssymb}
\usepackage{epsfig}
\usepackage{verbatim}

\IfFileExists{epsf.def}{\input epsf.def}{\usepackage{epsf}}

\IfFileExists{myowntimes.sty}{\usepackage[itboldsymbol]{myowntimes}\usepackage{mathrsfs}}
	 {\usepackage{times}\IfFileExists{mathrsfs.sty}{\usepackage{mathrsfs}}{\newcommand{\mathscr}{\mathcal}}}
\DeclareFontFamily{OT1}{eusb}{} \DeclareFontShape{OT1}{eusb}{m}{n} {<5> <6> <7> <8> <9> <10> <11> <12> <14.4> eusb10}{}
\DeclareMathAlphabet{\eusb}{OT1}{eusb}{m}{n}

\DeclareFontFamily{OT1}{eusm}{} \DeclareFontShape{OT1}{eusm}{m}{n} {<5> <6> <7> <8> <9> <10> <11> <12> <14.4> eusm10}{}
\DeclareMathAlphabet{\eusm}{OT1}{eusm}{m}{n}

\DeclareFontFamily{OT1}{eufm}{} \DeclareFontShape{OT1}{eufm}{m}{n} {<5> <6> <7> <8> <9> <10> <11> <12> <14.4> eufm10}{}
\DeclareMathAlphabet{\mathfrak}{OT1}{eufm}{m}{n}

\DeclareFontFamily{OT1}{fraktura}{}
\DeclareFontShape{OT1}{fraktura}{m}{n} {<5> <6> <7> <8> <9> <10> <11> <12> <13> <14.4> [1.1] eufm10}{}
\DeclareMathAlphabet{\fraktura}{OT1}{fraktura}{m}{n}

\DeclareFontFamily{OT1}{cmfi}{} \DeclareFontShape{OT1}{cmfi}{m}{n} {<5> <6> <7> <8> <9> <10> <11> <12> <13> <14.4> [0.9] cmfi10}{}
\DeclareMathAlphabet{\cmfi}{OT1}{cmfi}{b}{n}

\DeclareFontFamily{OT1}{cmss}{} \DeclareFontShape{OT1}{cmss}{m}{n} {<5> <6> <7> <8> <9> <10> <11> <12> <13> <14.4> cmss10}{}
\DeclareMathAlphabet{\cmss}{OT1}{cmss}{m}{n}

\newtheoremstyle{thm}{1.5ex}{1.5ex}{\itshape\rmfamily}{} {\bfseries\rmfamily}{}{2ex}{}

\newtheoremstyle{def}{1.5ex}{1.5ex}{\rmfamily\sl}{} {\bfseries\rmfamily}{}{2ex}{}

\newtheoremstyle{rem}{1.3ex}{1.3ex}{\rmfamily}{} {\itshape}
{} {1.5ex}{}

\theoremstyle{thm}
\newtheorem{theorem}{Theorem}[section]
\newtheorem{lemma}[theorem]{Lemma}
\newtheorem{proposition}[theorem]{Proposition}

\newtheorem*{Main Theorem}{Main Theorem.}

\theoremstyle{def}
\newtheorem{definition}[theorem]{Definition}

\theoremstyle{rem}
\newtheorem{remark}[theorem]{{\itshape Remark}}

\numberwithin{equation}{section}

\renewcommand{\section}{\secdef\sct\sect}
\newcommand{\sct}[2][default]{\refstepcounter{section}
\addcontentsline{toc}{section}
{{\tocsection {}{\thesection}{\!\!\!\!#1\dotfill}}{}}
\vspace{0.7cm}
\centerline{ 
\scshape\arabic{section}.\ #1} \nopagebreak \vspace{0.2cm}}
\newcommand{\sect}[1]{
\vspace{0.4cm} \centerline{\large\scshape\rmfamily #1}
\vspace{0.2cm}}

\renewcommand{\subsection}{\secdef\subsct\sbsect}
\newcommand{\subsct}[2][default]{\refstepcounter{subsection}
\addcontentsline{toc}{subsection}
{{\tocsection{\!\!}{\hspace{1.2em}\thesubsection}{\!\!\!\!#1\dotfill}}{}}
\nopagebreak\vspace{0.45\baselineskip} {\flushleft\bf
\arabic{section}.\arabic{subsection}~\bf #1.~}
\\*[3mm]\noindent
\nopagebreak}
\newcommand{\sbsect}[1]{\vspace{0.1cm}\noindent
\textbf{#1.~}\vspace{0.1cm}}

\renewcommand{\subsubsection}{%
\secdef \subsubsect\sbsbsect}
\newcommand{\subsubsect}[2][default]{%
\refstepcounter{subsubsection}
\addcontentsline{toc}{subsubsection}{{\tocsection{\!\!}
{\hspace{3.05em}\thesubsubsection}{\!\!\!\!#1\dotfill}}{}}
\nopagebreak
\vspace{0.15\baselineskip} \nopagebreak {\flushleft\rmfamily
\itshape\arabic{section}.\arabic{subsection}.\arabic{subsubsection}
\ \rmfamily #1\/.}\ }
\newcommand{\sbsbsect}[1]{\vspace{0.1cm}\noindent
\rmfamily \itshape
\arabic{section}.\arabic{subsection}.\arabic{subsubsection} \
\sffamily #1\/.\ }

\iffinal

\else

\fi

\renewcommand{\caption}[1]{%
\vglue0.5cm
\refstepcounter{figure}
\begin{minipage}{0.9\textwidth}\small {\sc Figure~\thefigure. }#1\end{minipage}}

\newcommand{\dist}{\operatorname{dist}}

\newcommand{\esssup}{\operatorname{esssup}}

\newcommand{\Int}{{\text{\rm Int}}}
\newcommand{\Ext}{{\text{\rm Ext}}}
\newcommand{\Sp}{{\text{\rm sp}}}
\newcommand{\Clean}{{\text{\rm clean}}}

\newcommand{\textd}{\text{\rm d}\mkern0.5mu}

\newcommand{\Var}{\text{\rm Var}}
\newcommand{\Cov}{\text{\rm \Cov}}

\newcommand{\splus}{\sigma^{\ell_<, +}}
\newcommand{\sminus}{\sigma^{\ell_<, -}}

\renewcommand{\AA}{\mathcal A}
\newcommand{\BB}{\mathcal B}
\newcommand{\CC}{\mathcal C}
\newcommand{\DD}{\mathcal D}

\newcommand{\FF}{\mathcal F}

\newcommand{\KK}{\mathcal K}

\newcommand{\OO}{\mathcal O}

\newcommand{\N}{\mathbb N}

\newcommand{\BbbP}{\mathbb P}

\newcommand{\R}{\mathbb R}

\newcommand{\Z}{\mathbb Z}

\newcommand{\scrH}{\mathscr{H}}

\def\myffrac#1#2 in #3{\raise 2.6pt\hbox{$#3 #1$}\mkern-1.5mu\raise 0.8pt\hbox{$#3/$}\mkern-1.1mu\lower 1.5pt\hbox{$#3 #2$}}

\newcommand{\hate}{\hat{\text{e}}}
\newcommand{\ra}{\rightarrow}

\author[N. Crawford]{Nicholas Crawford}
\thanks{{\tt
    email:nickcrawford12345@gmail.com}, Supported in part at the Technion by an Marilyn and Michael Winer Fellowship}

\begin{document}

\title[Random Field Induced Ordering]{On Random Field Induced Ordering in the Classical $XY$ Model}
\thanks{\hglue-4.5mm\fontsize{9.6}{9.6}\selectfont\copyright\,\textbf{Sept.  1 2010} by \textbf{Nick Crawford}.
Reproduction, by any means, of the entire
article for non-commercial purposes is permitted without charge.\vspace{2mm}}
\maketitle

\vspace{-5mm}
\centerline{\textit{Department of Industrial Engineering, The Technion, Haifa, Israel}}

\vspace{2mm}
\begin{quote}
\footnotesize \textbf{Abstract:}
Consider the classical $XY$ model in a weak random external field pointing along the $Y$ axis with strength $\epsilon$.  We study the behavior of this model as the range of the interaction is varied.  We prove that in any dimension $d \geq 2$ and for all $\epsilon$ sufficiently small, there is a range $L=L(\epsilon)$ so that whenever the inverse temperature $\beta$ is larger than some $\beta(\epsilon)$, there is strong residual ordering along the $X$ direction.  
\end{quote}
\vspace{2mm}

\section{Introduction}
In this paper we study an interesting phenomenon which has received little attention  in the mathematical physics literature:  random field induced ordering. The paradigmatic example is given by the following  classical $XY$ spin system.  Let $\{\alpha_z\}_{z\in \Z^d}=\{\alpha_z(\omega)\}_{z\in \Z^d}= $ be a family of i.i.d. $\{\pm1\}$ valued random variables taking each value with equal probability. Configurations are denoted by $\omega \in \Omega=\{0, 1\}^{\Z^d}$ with $(\Omega, \BB, \mathbb P)$ the corresponding probability space.
Let $\epsilon>0$ be fixed.
The (random) Hamiltonian of the model of interest is
\begin{equation}
\label{Eq:Ham}
\scrH(\sigma)=\scrH^{\omega}(\sigma) = -  \sum_{\langle z, z'\rangle } \sigma_z \cdot \sigma_{z'} - \epsilon \sum_{z \in \Lambda_N} \alpha_z(\omega) \hate_2 \cdot\sigma_z 
\end{equation}
where the sum is over the set of nearest neighbor bonds in $\Z^d$.  Here for each $z \in \Z^d$, $\sigma_z \in \mathbb S^1$ and we take as \textit{a priori} measure $\textd \nu(\sigma_z)$ Haar measure on $\mathbb S^1$.  Finally $\hate_1, \hate_2$ denote the two vectors $(1, 0),\: (0, 1)$ respectively.

Superficially, \eqref{Eq:Ham} is at the intersection of a number of notable $d=2$ phenomena:  On the one hand, if $\epsilon=0$, we have a pure $XY$ model, which does not order (magnetically) at any temperature as a consequence of the Mermin-Wagner theorem \cite{MW}.  On the other, if the spin space is take to be $\{-\hate_2, \hate_2\}$ rather than $\mathbb S^1$, then the model is the random field Ising model (RFIM), which does not order either as explained heuristically by Imry-Ma \cite{IM} and demonstrated rigorously by Aizenman-Wehr \cite{AW}.

The question then is what happens when the two models are "combined"?  In the (physics) literature it is expected, surprisingly at first, that there is in fact magnetic ordering at low temperature for any $d \geq 2$ and, at present, no truly rigorous mathematical results exist (in any dimension).  The $d=2$ case is particularly subtle for a number of reasons.  An even more surprising feature of this model is the direction of the ordering, expected to occur along the $\hate_1$ direction.

Interest in systems of this type have seen a flurry of activity in the physics literature in the last few years; a very recent review is \cite{SPL-Nature} which focuses on these effects in the quantum context, see also \cite{Ahor, Wehr-et-al-1, Wehr-et-al-2} for work on classical systems.  A related mechanism, not covered in \cite{SPL-Nature},  appears in \cite{ALL}.    There, the interesting point is that the magnetic field is determined by a random gauge potential.  As a result the average net field strength in a cube of side length $L$ scales as $\sqrt{L}$ in two dimensions as opposed to the scaling of $L$ in \eqref{Eq:Ham}.

Since we have so far been unable to treat directly the nearest neighbor model, to gain understanding of the mechanism behind the ordering we turned to its mean field theory.   Some light is shed in this simplification as we find that the free energy landscape of the model has a double well, with minima at (roughly) $\pm \rho(\beta) \hate_1$ for some positive $\rho(\beta)$ tending to $\sqrt{1-\epsilon^2}$ as $\beta \ra \infty$. 
This observation is the starting point for analysis of the system \eqref{E:Gibbs} when we take lattice models with Kac interaction of range $L = Poly(\frac 1\epsilon)$.  The result is formulated as Theorem \ref{T:Main}.  No attempt has been made to optimize the degree of this polynomial, although the discussion which follows indicates that a lower bound of $L \geq \epsilon^{-\frac 2d}$ is necessary to approximate lattice systems by mean field theories.

We should note that the mean field theory was previously addressed in \cite{Ahor, Wehr-et-al-1} in somewhat different ways.  Those treatments are not sufficient for our purposes and moreover,   the authors focus on critical behavior in the $(\epsilon, \beta)$ plane.  As we will argue, the link between this and the behavior of lattice systems seems tenuous, at least in the weak field regime.   

To provide context to the free energy functional will we introduce in a moment consider the formulation of \eqref{Eq:Ham} on the complete graph on $N$ vertices:
the Hamiltonian becomes
\begin{equation}
\label{Eq:Ham-C}
\scrH(\sigma) = - \frac{1}{N} \sum_{z, z'\in [N]} \sigma_z \cdot \sigma_{z'} - \epsilon \sum_{z \in [N]} \alpha_z \hate_2 \cdot\sigma_z.
\end{equation}
It is not difficult to compute $\phi$, the negative of the large deviation rate function for the vector observable $(M_N^+, M_N^-)$, where $M_N^{\pm}$ denotes the spatial average spins $\sigma_z$ over $\{z: \alpha_z= \pm 1\}$ respectively.
For $m \in \R^2$, let $S(m)= \inf_{h \in \R^2} (G(h)-  m \cdot h) $ with $G(h)= \log \int_{\mathbb S^1} \textd \nu(\sigma) e^{\sigma \cdot h}$ denoting the log moment generating function of $\nu$.  Let
\[
\phi(m^+, m^-)=\phi_{\beta, \epsilon}(m^+, m^-):= -\frac{1}{2}\|\bar m\|_2^2 - \frac{\epsilon}{2}\hate_2\cdot (m^+-m^-) - \frac{1}{2\beta}(S(m^+) + S(m^-))
\]
where $\bar m= \frac 12 (m^+ + m^-)$ and $f(m^+, m^-)= \phi(m^+,m^-)- \inf_{(m^+, m^-)} \phi(m^+,m^-)$.  
We use the notation
\[
\|(m^+, m^-)- (m^+_{0}, m^-_{0})\|= \max(\|m^+- m^+_{0}\|_2, \|m^-- m^-_{0}\|_2)
\]
with $\|\cdot\|_2$ denoting the usual Euclidean norm.
Then
\[
\lim_{\delta \ra 0} \lim_{N \ra \infty} \frac{-1}{N} \log \langle \mathbf 1_{\{\|(m^+, m^-)- (M^+_{N}, M^-_{N})\|< \delta\}} \rangle_N^{\omega} = f(m^+, m^-)
\]
where $\langle \cdot \rangle_N^{\omega}$ is the Gibbs state associated to \eqref{Eq:Ham-C} and convergence occurs for all $(m^+, m^-)$ $\omega$-a.s.

The main properties of $\phi$ are summarized in the following theorem.
\begin{theorem}[Low Temperature Mean Field Phase Diagram]
\label{T:MFT}
There exists $\epsilon_0>0$ and $\beta_0>0$ so that for all $\epsilon<\epsilon_0, \beta>\beta_0$, the free energy functional $f(m^+, m^-)$ has precisely two minimizers $\pm({\bf m^+, \bf m^-})$.  These minimizers are characterized by
\[
\|{\bf m}^+\|_2= \|{\bf m}^-\|_2=  \rho,
\]
\[
{\bf m^+} \cdot \hate_1 = {\bf m}^-\cdot \hate_1= \cos(\theta).
\]
where,
$\rho \leq 1$
and $\theta \in [\pi/2, \pi/2]$ satisfy the mean field equations
\begin{align*}
&\sin(\theta)= \frac{ \epsilon}{\rho},\\
&\rho = \frac{1}{\beta} \partial_{\rho} S(\rho \hate_1)  \quad (\text{ the maximal solution, which is nonzero}).
\end{align*}
In particular, $|\rho|$ is bounded away from $0$ for all $\beta$ sufficiently large and $\epsilon$ small and consequently $\theta=O(\epsilon)$.

Further, we have the following stability estimate:
\begin{multline}
  \phi(m^+, m^-)-\phi(\mathbf m^+, \mathbf m^-) \geq \\
  c(\epsilon_0, \beta_0)\|(m^+, m^-)+(\mathbf m^+, \mathbf m^-)\|\wedge \|(m^+, m^-)-(\mathbf m^+, \mathbf m^-)\|   \wedge \epsilon^2
\end{multline}
with all (non-minimizing) stationary points $(m^+_0, m^-_0)$ satisfying
\[
  \phi(m^+_0, m^-_0)-\phi(\mathbf m^+, \mathbf m^-) \geq \frac{\epsilon^2}{2}.
\]
\end{theorem}

The above theorem says that free energy $\phi$ has exactly two minimizers and that the height of the minimax barrier between the two minimizers is of order $\frac{\epsilon^2}{2}$ uniformly in $\beta$ large.  It also says that these two minimizers are \textit{transverse to the direction of the randomness} so that fluctuations of the local fields  do not favor one of these minima over the other.  This gives a stability not present in the RFIM.  

To make the mechanism for ordering even more transparent, we may attempt to interpolate between the random field transverse field $XY$ model and the RFIM by considering $\Z_n$ Clock Models.   The spin space is the set of $n$th roots of unity on the unit circle, viewed as vectors in $\R^2$, with Hamiltonian as in \eqref{Eq:Ham}.
If $n=4$, the randomness forces the system to behave much as in the case of the RFIM, as can be seen purely from the consideration of ground states. 
One may then wonder what happens as $n$ is increased.  We find, at least for the ground states, that for each $\epsilon>0$, there is  a crossover behavior from RFIM to the random transverse field $XY$ model occurring at $n \sim \frac{1}{\epsilon}$.  That is, if $n<< \frac 1\epsilon$ then ground states oscillate between $\pm \rho \hate_2$ where as if $n>>\frac 1\epsilon$, there are a finite number  ground states, stable under the noise, all roughly parallel to $\hate_1$. In particular, in contrast to statements made in \cite{Wehr-et-al-1}, a spin space with truly continuous symmetry is not necessary for the effect of interest to occur.
All that matters is that the spin space has enough freedom to take advantage of local fluctuations.  

We would also like to point out that one can perform a mean field analysis in case $\mathbb P(\alpha_x =1) = p \neq \frac 12$ and a similar picture emerges in the mean field theory.  In particular letting $q=1-p$, the mean field equations become
\begin{align*}
p{\bf m^+} \cdot \hate_1  &= q {\bf m}^-\cdot \hate_1\\
p{\bf m^+} \cdot \hate_2-  q {\bf m}^-\cdot \hate_2&= \epsilon \\
\rho_{+} &=  \frac{1}{2 p \beta} \partial_{\rho} S(\rho_+ \hate_1)  \quad (\text{maximal solution})\\
\rho_{-} &=  \frac{1}{2 q \beta} \partial_{\rho} S(\rho_- \hate_1) \quad (\text{maximal solution})
\end{align*}
as can be seen by following the proof of Theorem \eqref{T:MFT}.  These equations imply that at low temperatures, the picture presented above in the unbiased case persists--there are two minimizers, symmetric with respect to the $Y$-axis.  The location of the average of the two components of a minimizer has a non zero $Y$ component with sign and magnitude  determined by the bias.

Let us return to requirement $L = Poly(\frac 1\epsilon)$.
Fluctuations of the local fields are of order $\sqrt{N}$ typically, and can change the finite volume low temperature free energy landscape by $O(\frac{\epsilon}{\sqrt{N}})$.  Thus the most likely order parameter values at finite volume can have $\phi$-free energy which differs from the absolute minimum of $\phi$ by an amount of order $\frac{\epsilon}{\sqrt{N}}$.  Large deviations implies that the height barrier in a system of $N$ vertices is of order $\epsilon^2 N$, so the effects of field fluctuations on macroscopic observables will only be suppressed if $ N\epsilon^2>> \epsilon \sqrt{N}$.  This leads us to a fundamental requirement for the validity of the $\phi$-mean field picture \text{even to the model on the complete graph}: $\sqrt{N} >> \frac{1}{\epsilon}$.  Translated to Kac interactions, this means we must take the range of our interaction $L$ so that $L^{d/2}>> \frac{1}{\epsilon}$. 

The above discussion provides a first noteworthy point of our work.   Though it is not believed that randomness induced ordering depends on taking a long range interaction, our work emphasizes the difficulties in drawing conclusion about short range models (even in high dimension) by extrapolating from mean field analysis, particularly when the random field strength is weak.  This point seems to have been overlooked in the literature.

A second main point (somewhat counter to the first) is that we are able to give a rigorous example of a lattice system where the combination of  a random field acting in one direction with a coupling that has continuous symmetry disrupts the behavior of the two $d=2$ systems discussed above (Mermin-Wagner and RFIM).  Indeed, it is known that neither of those models order at low temperature \textit{even when the interaction range is spread out as $Poly(\frac 1\epsilon)$}.

The rest of the paper is organized as follows.  In the next section, we precisely formulate our main result Theorem \ref{T:Main}: the existence of residual magnetization along the $\hate_1$ direction for Kac models with sufficiently long range interaction.  Section \ref{S:Contours} states the main Lemmata needed for  the proof of the Theorem \ref{T:Main} and  on their basis provides a proof of the Theorem.  Section \ref{S:MFT} is devoted to the proof of Theorem \ref{T:MFT} and other technical estimates needed regarding the mean field theory.  There after, the paper is devoted to a justification of the Lemmata appearing in Section \ref{S:Contours}.  The techniques used are mostly taken from the book \cite{Pres-Book}, with modifications needed to treat the randomness and the fact that we are working with a continuous spin space.  This means making appropriate definition of course grained contours and comparing contour energies to a certain free energy functional evaluated on deterministic magnetization profiles which are defined on the support of each contour.

\section{Main Result}
Let $\mathcal S= (\mathbb S^1)^{\Z^d}$ be endowed with the product topology and associated Borel $\sigma$-field $\BB_0$.  For any $\sigma \in \mathcal S$ and any set $ \Lambda \subset \Z^d$ let $\sigma_{\Lambda}$ denote the restriction of $\sigma$ to $\Lambda$ and $\textd \nu_{\Lambda}=\textd \nu(\sigma_{{\Lambda}})$ denote the corresponding restriction of $\nu$. $\mathcal S_{\Lambda}$, $\mathcal B_{0, \Lambda}$ will denote the associated space of spin configurations and sigma fields respectively.

Let $\Lambda_{N}=\{z \in \Z^d: \|z\|_\infty \leq N\}$ where $\|z\|_\infty$ denotes the $\ell^{\infty}$ length in $\R^d$. 
and introduce a length scale $L$ which represents the range of the interaction.  So $J_L(z, z'):= c_L J(
\|z-z'\|_2/L)$ where $J$ is $C^{\infty}_c(\{x \in \R^d: \|x\|_2 \leq 1\})$, $c_L^{-1}= \int \textd x J(\|x\|/L)$ is a normalizing constant.
Let $\epsilon>0$ be fixed.
Given $\sigma_{\Lambda_N^c} \in \mathcal S_{\Lambda_N^c}$, the random Hamiltonian for our model is
\[
\scrH_N(\sigma_{\Lambda_N}|{\sigma_{\Lambda_N^c}}) = -  \sum_{z,z' \in \Lambda_N} J_L(z, z') \sigma_z \cdot \sigma_{z'} - \epsilon \sum_{z \in \Lambda_N} \alpha_z \hate_2 \cdot\sigma_z - \sum_{z\in \Lambda_N, z' \in \Lambda_N^c} J_L(z, z') \sigma_z \cdot \sigma_{z'}
\]
Notice the crucial property that, even in the presence of the random fields, the energy $\scrH_N(\sigma_{\Lambda_N}|{\sigma_{\Lambda_N^c}})$ is invariant with respect to simultaneous reflection of all spins about the $Y$-axis.

The Gibbs-Boltzman probability distribution associated with this Hamiltonian is defined by Radon-Nikodym derivative relative to $\nu_{\Lambda_N}$:
\begin{equation}
\label{E:Gibbs}
\textd\mu_{N}^{\sigma_{\Lambda_N^c}}(\sigma_{\Lambda_N})= \textd \mu_{N}^{\sigma_{\Lambda_N^c}, \omega}(\sigma_{\Lambda_N}) \propto e^{-\beta \scrH_N(\sigma_{\Lambda_N}| {\sigma_{\Lambda_N^c}})}\textd \nu(\sigma_{{\Lambda_{N}}})
\end{equation}
with constant of proportionality $Z_N^{\omega}(\beta, \sigma_{\Lambda_N^c})$ the (random) partition function of the system.

\noindent
\textbf{Horizontal Boundary Conditions:}
For convenience, let $\mu_{N}^{\ra, \omega}$ be the random Gibbs state with $\rightarrow$
denoting boundary conditions given by setting all boundary spins equal $\bar{\mathbf m}:= \frac 12(\bf m^++ \bf m^-)$.  Even though this choice is not strictly in the spin space, the model still makes sense.  Further, this external configuration can be effectively produced mesoscopically by taking spins to be of the form $a \hate_1 \pm b \hate_2$ where the sign varies according to the parity of the underlying lattice site.  In this way one obtains block average magnetizations $a \hate_1 + O(k^{-2})$ where $k$ is the size of the box.  Since we are considering Kac interactions, the effect is the same as our Horizontal Boundary conditions to within an error which plays no significant role.

Given $\lambda>0$, let $\ell_>= L^{1+\lambda}$.  Let $B^{\ell_>}_z= \{ x \in \Z^d: \|z-x\|_{\infty}\leq \ell_>\}$.  A subset $\Lambda$ of $\Z^d$ will be said to be $\ell_>$- measurable if $\Lambda$ is a union of blocks $B^{\ell_>}_r$ so that $r \in (2\ell_>+1) \Z^d$.
Let
\[
B^{\pm , \ell_>}_z= \{ x \in B^{\ell_>}_z: \alpha_x=\pm1\}\
\]
We define the block average magnetizations
\begin{align*}
&M^{\pm}_z = \frac{1}{|B^{\pm , \ell_>}_z|}\sum_{ x \in B^{\pm , \ell_>}_z} \sigma_x,\\
&M_z = \frac{1}{|B^{ \ell_>}_z|}\sum_{ x \in B^{ \ell_>}_z} \sigma_x.
\end{align*}

\begin{theorem}[Main Theorem]
\label{T:Main}
Let $d \geq 2$ be fixed.  We can find $\xi_0, \epsilon_0> 0$ and $\lambda>0$ so that for every $0<\epsilon< \epsilon_0$ and $0< \xi< \xi_0$, there exists $\beta_0(\epsilon, \xi), L_0(\epsilon, \xi, \lambda)>0$ for which the following holds:  If $\beta> \beta_0, L \geq L_0$,  then for almost every $\omega\in \Omega$, there exists an $\ell_>=\ell_>(L)$-measurable subset $\DD_{\omega} \subset \Z^d$ and an $N_0(\omega)\in \N$ such that:
\begin{enumerate}
\item
\[
|\DD \cap \Lambda_N| \leq q |\Lambda_N|
\]
for all $N \geq N_0(\omega)$.
\item
For each $z\in \Lambda_N$ with $B^{\ell_>}(z) \cap \DD = \varnothing$,
\[
\|\langle M^{\pm}_z \rangle_N^{\omega, \ra} - \mathbf m^{\pm}\|_2 \leq \xi
\]
and 
\[
\|\langle M_z \rangle_N^{\omega, \ra} - \bar{\mathbf m}\|_2 \leq \xi
\]
Explicitly, we require
\[
\epsilon^2 \wedge \xi^2 \wedge  L^{-\lambda d}  \geq C (L^{\frac 23 (\lambda-1)+2\lambda d} + L^{2\lambda d}e^{-c \epsilon^2 L^{\lambda}}),
\]
$L^{-\lambda}< \frac{1}{8}$ and $\xi>c L^{-\frac 54(1-\lambda)}$.
where $c, C$ are universal constants.
\end{enumerate}
\end{theorem}

\begin{remark}
Let us note that the proof we employ below will work, with the appropriate modifications, in the case of bias $\mathbb P(\alpha_z=1)= p \neq \frac 12$.
\end{remark}

\textbf{Notational Convention}
Below the constants $c, C, C_1>$ will always be universal in the sense that they only depend on $d$ and the sup norm of $\nabla J$ but not on $\xi, \epsilon, \beta, L$ etc.  Their values may (will) change from line to line.

\section{Course-Graining and Contours}
\label{S:Contours}
We define $\ell_{<}=\lfloor L^{1-\lambda}\rfloor$ and $\ell_{>}= \lceil L^{1+\lambda}\rceil$  to be two scales respectively slightly smaller and slightly larger than $L$.
We shall assume $2\ell_{<}+1$ divides $L$ and $L$ divides $2\ell_{>}+1$.

The scales $\ell_{<}$ and  $\ell_{>}$ introduce a filtration of blocks in $\R^d$ and, by taking intersections, in $\Z^d$.  We shall say that a block $B_r= \{x \in \R^d : \|x-r\|_{\infty} \leq L'\}$ is measurable relative to the scale $L'$ if $r \in (2L'+1) \Z^d$. 
For $\Lambda \subset \Z^d$ finite, let
\[
N^{\pm}_{\Lambda}=\{\alpha_x: x\in \Lambda,  \alpha_x=\pm 1\}.
\]
It is standard that
\[
\BbbP(|N^{\pm}_{\Lambda}|- |\Lambda|/2 \geq A \sqrt{|\Lambda|/2}) \leq 2 e(-A^2/4).
\]
where $|\Lambda|$ denote the cardinality of a finite subset of $\Z^d$.
For our purposes, it will suffice to take $A= |\Lambda|^{\kappa}$ for some $\kappa \in (0, 1/2)$. If $\Lambda= B_r \subset \Z^d$ for some $\ell_<$-measurable block $B_r$, we will use the shorter notation $N^{\pm}_{r}$.

We define
\[
\sigma^{\ell_<,\pm}_z= \frac{1}{|N^{\pm}_{r}|} \sum_{x \in B_r: \alpha_x = \pm 1} \sigma_x
\]
where $B_r$ denotes the $\ell_<$-measurable block containing $z$.  Note that this depends on the realization of the randomness.  Also let
\[
\sigma^{\ell_<}_z=\frac{1}{\ell_<^d} \sum_{x \in B_r} \sigma_x,
\]
which does not.
For $D \subset \R^d$ Borel measurable, we shall use the notation
\[
L^{\infty}(D)=\{ m: D \ra \R^2\: s.t.\: m\text{ is Borel measurable and }\esssup_{z\in D} \|m_z\|_2 < \infty \}
\]
where
$\|m_z\|_2$ is the usual Euclidean norm in $\R^2$ (more generally this notation is used for $2$-norms in all Euclidean spaces $\R^d$).  We will also use $\|m_z\|_\infty= \max_{1\leq i \leq d} |m_z(i)|$ for vectors $m_z \in \R^d$.  The induced norm on $L^\infty(D)$ is denoted by
\[
\|m\|_{L^\infty(D)}= \esssup_{z\in D} \|m_z\|_2.
\]

For any pair $(m^+, m^-) \in L^{\infty}(\R^d)\times L^{\infty}(\R^d)$, we denote $\bar m= \frac{1}{2}(m^+ + m^- )$ and for any $z\in \hat B_r$ with $B_r$ $\ell_<$-measurable
\[
m^{\ell_<,\pm}_z=  \frac{1}{\ell_<^d}\int_{B_r} \textd y \: m^{\pm}_y
\]
and
\[
\bar{m}^{\ell_<}_z = \frac{1}{2}( m^{\ell_<, +}_z + m^{\ell_<, -}_z ).
\]
Below, we refer to magnetization profiles $(m^+, m^-)$ which are piecewise constant over the set of all $B_r$   which are $\ell_<$-measurable (and in the domain of definition of $(m^+, m^-)$) as \textit{$\ell_<$-piecewise constant}.

Fix $\Lambda \subset \Z^d$ which is $\ell_>$-measurable. Given a spin configuration $\sigma \in \mathcal S$, and for any $z \in \Lambda, z \in B_r$ such that $B_r$ is $\ell_<$-measurable, we introduce the phase variables:
\begin{itemize}
\item
\[
\eta=\eta^{\xi}_z= \begin{cases} 1 \quad \text{ if $\|(\sigma^{\ell_<, +}_z, \sminus_z)-  ({\bf m^+}, {\bf m^-})\|_{\infty} \leq \xi$},\\
-1\quad \text{ if $\|(\splus_z, \sminus_z)+  ({\bf m^+}, {\bf m^-})\|_{\infty} \leq \xi$},\\
0 \quad \text{ otherwise}.
\end{cases}
\]
\item
For any $z \in \Lambda, z \in B_r$ with $B_r$ $\ell_>$-measurable
\[
\theta=\theta^{\xi}_z= \begin{cases} 1 \quad \text{ if $\eta^{\xi}_y =1$ for all $y \in B_r$},\\
-1\quad \text{ if \text{ if $\eta^{\xi}_y =-1$ for all $y \in B_r$}},\\
0 \quad \text{ otherwise.}
\end{cases}
\]
\item
These phase variables are extended to the set of boxes $B_r \subset \Lambda^c$ with $\dist(B_r, \Lambda)\leq \ell_>$ via
\[
\eta=\eta^{\xi}_z=\eta^{\xi}_{\Lambda, z}= \begin{cases} 1 \quad \text{ if $\| \sigma^{\ell_<}_z-  \bar{{\bf m}}\|_{\infty} \leq \xi$},\\
-1\quad \text{ if $\| \sigma^{\ell_<}_z + \bar{{\bf m}}\|_{\infty} \leq \xi$},\\
0 \quad \text{ otherwise}.
\end{cases}
\]
 and similarly for $\theta$.
We will always work under conditions in which $\theta \equiv \pm 1$ over connected components of $\Lambda^c$.

\item
For any $z \in \Lambda, z \in B_r$ with $B_r$ $\ell_>$-measurable
\[
\Theta=\Theta^{\xi}_z=\Theta^{\xi}_{\Lambda, z}= \begin{cases} 1 \quad \text{ if $\theta^{\xi}_y =1$ for all $y : \|z-y\|_{\infty} \leq \ell_{>}$}\\
-1\quad \text{ if $\theta^{\xi}_y =1$ for all $y : \|z-y\|_{\infty} \leq \ell_{>}$}\\
0 \quad \text{ otherwise.}
\end{cases}
\]
\end{itemize}
We emphasize that $\Theta$ depends on $\Lambda$ and the boundary conditions.
Another way of defining $\Theta$ is to say that an $\ell_>$-measurable block has nonzero value of $\Theta$ if that block and all neighbors  (in the $\ell^{\infty}$ sense), including boundary boxes, have same nonzero value.

For any $\ell_>$-measurable region $D$ in $\R^d$, we may extend all of these notions to deterministic magnetization profiles $(m^+_z, m^-_z) \in L^{\infty}(D) \times L^{\infty}(D)$.  Note that the definition of $\Theta$ requires the auxiliary input of a fixed boundary condition $\bar m_{0, z} \in L^{\infty}(D^c)$.  The boundary condition used will be clear from context.

For any set $A \subset \Z^d$ we can associate a subset $\hat A \subset \R^d$ which is the union of boxes of side length $1$ centered at the elements of $A$.  We shall say that $A$ is connected if $\hat{A}$ is (note this is NOT the same as connectivity in $\Z^d$).
$\hat{A}^c$ decomposes into one infinite connected component $\Ext(\hat{A})$ and a number of finite connected components $(\Int_i(\hat{A}))_{i=1}^{m}$ with $\Int(\hat{A})= \cup_{i=1}^m \Int_i(\hat{A})$.
Let us denote the $\ell_>$ enlargement of a set $A \subset \R^d$ by
\[
\delta(\hat{A})=\cup_{\{B_r \: \ell_>\text{-measurable}\: :\: \dist(\hat B_r, \hat{A})< \ell_{<}\}}\hat B_r
\]
where $\dist(\hat B_r, \hat{A})$ is the Hausdorff Distance between sets in $\R^d$ in the $\ell_{\infty}$ metric.
From here we may introduce $\delta(A), \Int_i(A), etc.$ by taking intersection of each defined set with $\Z^d$.
The closure of a set $A \subset \Z^d$ is defined to be $c(A)=\delta(A)\cup \Int(A)$.
It is standard that the connected components of  $R^+= \{z \in \Z^d: \Theta^{\xi}_z=1\}$ and $R^-= \{z\in \Z^d : \Theta^{\xi}_z=-1\}$ are separated by connected subsets of  $R^0= \{z\in \Z^d: \Theta^{\xi}_z=0\}$.

\begin{definition}
A \textit{contour} $\Gamma$ is defined to be the pair $(\Sp(\Gamma), \theta_{\Gamma})$ where $\Sp(\Gamma) \subset \Z^d$ is connected, $\ell_>$-measurable and $\theta_{\Gamma}(z)$ is an $\ell_<$-measurable $\{-1, 0, 1\}$-valued function on $\Sp(\Gamma)$ which gives the values of the phase specification on $\Gamma$. In the previous definitions, whenever the set $A$ in question happens to be $\Sp(\Gamma)$ we will write $\delta(\Gamma), c(\Gamma),$ etc.   Let $N_{\Gamma}= |\delta(\Gamma)|/(2\ell_{>}+1)^d$, so that $N_{\Gamma} \in \mathbb N$.  
\end{definition}
Let us introduce the notation $\delta_{ext}(\Gamma)=\delta(\Gamma)\cap \Ext(A)$ and $\delta^{i}_{in}(\Gamma)= \delta(\Gamma) \cap \Int_i(A)$.  These are evidently disjoint.  By definition of $\Gamma$, each of these sets is connected  (in our sense) and disconnected from the rest.

We shall denote by
\begin{multline}
\mathbb X(\Gamma)= \{ \sigma\in \mathcal S: \Gamma \text{ is a contour for $\sigma$}\}\\
:= \{\sigma: \Sp(\Gamma) \text{ is a maximal connected subset of } R^0(\sigma)\text{ and } \theta_z(\sigma) \equiv \theta_{\Gamma}(z)\text{ on } \Sp(\Gamma)\}.
\end{multline}
We shall say that $\Gamma$ is a contour for $\sigma$ if $\sigma \in \mathbb X(\Gamma)$.

By our definitions, specifying that $\Gamma$ is a contour of $\sigma$ lets us recover the values of $\Theta_z(\sigma), \theta_z(\sigma)$ on $\delta(\Gamma)$ (see \cite{Pres-Book} for details).  This convenient property allows us analyze systems of contours without worrying about the microscopic spin configuration far away from the contour.
Two contours $\Gamma_1, \Gamma_2$ are said to be compatible if $\delta(\Gamma_1)\cap \Sp(\Gamma_2)= \varnothing$ and $\theta_{\Gamma_1}= \theta_{\Gamma_2}$ on the domain of intersection of $\delta(\Gamma_1), \delta(\Gamma_2)$.

So far we have considered contours at the level of spin configurations.  We would like to be able to show that under certain finite volume Gibbs measures, a contour costs $e^{-c(\xi, \ell_{<}) N_{\Gamma}}$, the constant $c$ being made large by appropriate choice of $\beta, \xi , L$.  In general, such an estimate will NEVER be true uniformly in the presence of randomness, but we can hope that for the family of $\ell_{<}$-measurable blocks $\{B_r\}$, large fluctuations of the variables $N^{\pm}_{B_r}- |B_r|/2$ are sufficiently rare so that we can still extract $e^{-c'(\xi, \ell_{<}) N_{\Gamma}}$ in cost from MOST contours.

Thus we turn to the interplay between spin configurations and randomness.
We introduce (more) phase variables associated to the randomness -- $\omega\in \Omega$ -- restricted to boxes $B_r \subset \Z^d$ with side length $2\ell_<+1$.
Let $1/2>\kappa>0$ be fixed.  For $B_r$ $\ell_<$-measurable and $z \in B_r$ define
\[
\varphi_{z}= \varphi^{\kappa}_{\ell_<, z}(\omega):=
\begin{cases}
1 \quad \text{ if $|N^{\pm}_{B_r}- |B_r|/2| < |B_r|^{1/2 + \kappa}$}\\
0 \quad \text{ otherwise.}
\end{cases}
\]
and for any $z \in B_r$, with $B_r$ $\ell_>$- measurable,
\[
\Xi_{z}= \Xi^{\kappa}_{\ell_>, z}(\omega):=
\begin{cases}
1 \quad \text{ if $\varphi_y=1$ for all $y \in B_r$}\\
0 \quad \text{ otherwise.}
\end{cases}
\]

\begin{definition}
Let us say that an $\ell_>$-measurable set $S$ is \textit{($\kappa, p$)-clean}, or just clean, if
\[
S^{\text{clean}}:= \{z \in S: \Xi_{\ell_>, z}=1\}
\] satisfies $|S^{\Clean}|/|S| > 1-p$. Otherwise, $S$ is called dirty.
A contour $\Gamma$ will be called clean if $\delta(\Gamma)$ is clean.
\end{definition}

Let \[
\AA=\{Y\subset \Z^d: \:Y \text{ is $\ell_>$-measurable, connected,  and $(
\kappa, p)$ dirty}\}
\]
Let $\DD:= \cup_{Y\in \KK} c(Y)$.
Note that these definitions only depend on the realization $\omega$ and not on possible spin configurations.

Given a contour $\Gamma$, $\Gamma^*$ will denote the $(\Omega, \BB, \mathbb P)$ event
\begin{equation}
\label{E:1}
\Gamma^*=\{\delta(\Gamma) \text{ is clean and $c(\Gamma)$ is not strictly contained in $\DD$}\}
 \end{equation}
 If $\Gamma^*$ occurs, $\Gamma$ will be called a $*$-clean contour.
Given a spin configuration $(\sigma_{\Lambda_N}, \sigma_{\Lambda_N^c})$, let
\[
\mathbb X(\Gamma_1^*, \dotsc, \Gamma_m^*, \Gamma_{m+1},\dotsc, \Gamma_{m+n})= \cap_{i}  \mathbb X(\Gamma_i)
\]
where $ (\Gamma_1, \dotsc, \Gamma_m)$ satisfy the event defined in \eqref{E:1} and $( \Gamma_{m+1},\dotsc, \Gamma_{m+n})$ do not.  Otherwise we define the right hand side to be the empty set.
As a variation of standard definitions, let us say that
\[
(\Gamma_1^*, \dotsc, \Gamma_m^*, \Gamma_{m+1},\dotsc, \Gamma_{m+n}, \omega)
\]
are $*$-compatible if
\[
\mathbb X(\Gamma_1^*, \dotsc, \Gamma_m^*, \Gamma_{m+1},\dotsc, \Gamma_{m+n}) \neq \varnothing.
\]

\begin{lemma}
\label{L:Contours}
There exist $\delta,  \epsilon_0> 0$ so that  for any $p, \lambda, \kappa \in (0, \frac 13)$, $0< \xi< \delta$ so that $L^{-\lambda}< \frac{1}{8}$ and $\xi>c L^{-\frac 54(1-\lambda)} \log L$ and if $\epsilon< \epsilon_0$, there exists $\beta_{\epsilon}$ so that if $\beta> \beta_{\epsilon}$ then the following holds:

\noindent
Let $N$ be fixed and consider the event $\mathbb X(\Gamma_1^*, \dotsc, \Gamma_m^*, \Gamma_{m+1},\dotsc, \Gamma_{m+n})$ with $\Sp(\Gamma_i) \subset \Lambda_N$.  Then
\[
\mu_{\Lambda_N}^{\ra}(\mathbb X(\Gamma_1^*, \dotsc, \Gamma_m^*, \Gamma_{m+1},\dotsc, \Gamma_{m+n})) \leq e^{-q \sum_{i=1}^m N_{\Gamma_i}}
\]
where
\begin{multline}
\label{E:ContBon}
q =C_1 \beta L^{d(1-\lambda)} \Big\{ \epsilon^2 \wedge \xi^{2} \wedge L^{-\lambda d} - \\
C(L^{\lambda-1+2\lambda d}+\epsilon p L^{2\lambda d} +  \epsilon L^{(-d/2+d\kappa)(1- \lambda) +2\lambda d} + L^{2 \lambda d}e^{- c \epsilon^2 L^{\lambda}} +\beta^{-1} L^{2 \lambda d- \frac 54(1- \lambda)} \log L)\Big\}
\end{multline}
with $c, C_1, C>0$ universal constants.
\end{lemma}
This lemma is proved in several parts below.  Theorem \ref{T:Main} is then completed via the following Peierls contour counting argument.

\begin{proof}[Proof of Theorem \ref{T:Main}]

Fix $x \in \Z^d$.
Throughout this proof, let $B(x)$ denote the $\ell_>$-measurable block containing $x$.  Consider the collection of bounded $\ell_>$-measurable connected subsets  of $\Z^d$ containing $x$: $\{Y\}_{ Y \ni x}$. For each such $Y$, we first estimate the event that $Y$ is $(\kappa, p)$-dirty.
We have:
\[
\mathbb P(Y \text{ is ($\kappa, p$)-dirty})\leq 2^{N_Y} e^{-c\ell_<^{2\kappa d} p N_{Y}}.
\]
where $N_Y$ is the number of $\ell_>$-measurable blocks in $Y$.
Now the number of $\ell_>$-measurable connected sets $Y$ containing $x$ with $N_Y=r$ is well known to have the asymptotic $a_0^r$ for some fixed, dimension dependent constant $a_0$.
Thus
\[
\mathbb P(B(x) \text{ is in some ($\kappa, p$)-dirty $Y$})\leq \sum_{r \geq 1} (2a_0)^{r} e^{-c\ell_<^{2\kappa d} pr}.
\]
Modifying the estimate slightly, via the discrete isoperimetric inequality
\begin{equation}
\label{E:ClustBound}
\mathbb P(B(x) \text{ is in $c(Y)$ for some  ($\kappa, p$)-dirty $Y$})\leq C\sum_{r \geq 1} r^{d/(d-1)}(2a_0)^{r} e^{-c\ell_<^{2\kappa d} pr}.
\end{equation}
where $C$ is a universal constant coming from the isoperimetric bound.

Next we need a correlation bound. Let $A(x)= \{ B(x) \text{ is in $c(Y)$ for some  ($\kappa, p$)-dirty $Y$}\}$.  Using the fact that the events  $\{c(Y_i) \text{ is } (\kappa, p)-\text{dirty}\}$ are independent if $c(Y_1) \cap c(Y_2) = \varnothing$ we have
\begin{multline}
\label{E:Cor}
\mathbb P(A(x_1), A(x_2))-\mathbb P(A(x_1)) \mathbb P(A(x_2))
 \leq \\
 \mathbb P(B(x_1), B(x_2) \text{ are in $c(Y_1), c(Y_2)$ for some  ($\kappa, p$)-dirty $Y_1, Y_2$ with $c(Y_1) \cap c(Y_2) \neq \varnothing$}).
 \end{multline}
Then, estimating as in \eqref{E:ClustBound},
the right hand side is bounded by $C \ell_>^d e^{-c\ell_<^{2\kappa d} p \dist_{\ell_>}(x, y)}$
where $\dist_{\ell_>}(x, y)$ denotes the minimal number of blocks in an $\ell_>$-measurable block path from $B(x)$ to $B(y)$.

These bounds imply a constraint on parameters: namely
\[
L^{2\kappa d(1- \lambda)} p \geq c \log a_0
\]
We will assume from now on that $\kappa = \frac{1}{5}$, $p=L^{-\frac{d}{3}(1- \lambda)}$ and $\lambda< \frac 13$, so that the inequality holds for all $L>L_0$ for some $L_0 \in \N$.

Recall that
\begin{align*}
&\AA=\{Y\subset \Z^d:  \: Y \text{ is } (\kappa, p)\text { dirty, } \text{$\ell_>$-measurable and connected}\},\\
&\DD= \cup_{Y \in \AA}c(Y).
\end{align*}
Letting  $\mathcal D_N=\DD\cap \Lambda_N $, $\DD_N = \cup_{x \in \Lambda_N}A(x)$ so $\mathbb E[ |\mathcal D_N|] \leq C |\Lambda_N| e^{-c\ell_<^{2\kappa d} p}$
 and, via \eqref{E:Cor}, 
\[
\Var( |\mathcal D_N|) \leq C \sqrt{|\Lambda_N|}
\]
By taking $N=2^{k}$ applying Chebyshev's inequality to estimate deviations of $|\mathcal D_N|$ and then the Borel-Cantelli lemma along this subsequence,
we have that, for almost every $\omega\in \Omega$, there is $N_0(\omega)\in \N$ so that
\[
 \frac{|\DD_N|}{|\Lambda_N|}\leq C \ell_>^d e^{-c\ell_<^{2\kappa d} p} \leq CL^{2d}e^{-cL^{\frac{2}{45}}} 
\]
for all $N \geq N_0(\omega)$.
We have used here that $\mathcal D_N$ increases with $N$.

The rest of the proof is an application of Lemma \ref{L:Contours}.
As $p=L^{-\frac{d}{3}(1- \lambda)}$ we may now assume the dominant contributions to $qL^{-(1-\lambda)d}$ are
\[
  \beta (-\epsilon^2\wedge  \xi^2 \wedge  L^{-\lambda d}  + L^{\frac 23 (\lambda-1)+2\lambda d} + L^{2\lambda d} e^{-c \epsilon^2 L^{\lambda}})+L^{2\lambda d}.
\]

We thus require
\[
\epsilon^2 \wedge \xi^2 \wedge  L^{-\lambda d}  \geq 2 (L^{\frac 23 (\lambda-1)+2\lambda d} + L^{2\lambda d}e^{-c \epsilon^2 L^{\lambda}})
\]
which implies $q \geq  c_1 L^{d(1-\lambda)} \beta\epsilon^2 \wedge \xi^2 \wedge  L^{-\lambda d} -c_2 L^{2\lambda d}$.  
Clearly for each $d$ we may find $\lambda=\lambda(d)$ and then $L= L(\lambda, \epsilon, \xi)$ so that this holds for all $L> L(\lambda, \epsilon, \xi)$. 
Then we see that $q$ can be made arbitrarily large by the appropriate choice of $\beta$ with the rest of the parameters fixed.

Fix $x \in \Lambda_N$ so that $B^{\ell_>}(x) \cap \DD = \varnothing$ and consider the event
\[
\{\theta_{B^{\ell_>}(x)} \neq 1\}
\]
which is a subset of $\mathcal S_{\Lambda_N}$.  By definition of $\theta, \Theta$ and the $\ra$ boundary condition, there exists a largest contour $\Gamma$ so that $B^{\ell_>}(x) \subset c(\Gamma)$.
Moreover, since $B^{\ell_>}(x) \cap \DD = \varnothing$, $\Gamma$ must be $(\kappa, p)$-clean.
Decomposing $\{\theta_{B^{\ell_>}(x)} \neq 1\}$ into disjoint subsets according to this largest contour we have:
\[
\mu_N^{\ra, \omega}(\theta_{B^{\ell_>}(x)} \neq 1) \leq \sum_{\Gamma, \omega\: * \text{-compatible}: \: B^{\ell_>}(x) \subset c(\Gamma)} \mu_N^{\ra, \omega}(\mathbb X(\Gamma^*)).
\]
By Lemma  \ref{L:Contours},
\[
\mu_N^{\ra, \omega}(\theta_{B^{\ell_>}(x)} \neq 1) \leq C \sum_{r \geq 1} r^{d/(d-1)} (2a_0)^{r} e^{-q r} \leq C_1e^{-\frac{q}{2}}
\]
as long as $q> 2\log(2a_0)$.  The theorem now follows easily.
\end{proof}

\section{Mean Field Theory}
\label{S:MFT}
\begin{proof}[Proof of Theorem \ref{T:MFT}]
Let us use the polar coordinate parametrization $m^{\pm}= (\rho_{\pm}, \theta^{\pm})$.  Although there are cleaner ways of deriving what we need, as kindly pointed out by D. Ioffe, we will stick with direct computations in these coordinates.

Evidently all minimizers have the property that either $\theta^{\pm} \in [-\pi/2, \pi/2]$ or $\theta^{\pm} \in [\pi/2, 3\pi/2]$.  By symmetry, it is enough to treat the case $\theta^{\pm} \in [-\pi/2, \pi/2]$.  We may also assume $\theta^+ > \theta^-$.

In polar coordinates
\begin{multline}
\phi= \phi(\rho_{\pm}, \theta^{\pm})=-\frac{1}{8}(\rho_+^2 + \rho_-^2 + 2 \rho_+\rho_-\cos(\theta^+-\theta^-)) - \frac{\epsilon}{2}(\rho^+\sin(\theta^+) - \rho_-\sin(\theta^-)) \\
- \frac{1}{2\beta} (S(\rho_+\hate_1) + S(\rho_- \hate_1)\\
=\frac{1}{2}( g_{\beta}(\rho_+) + g_{\beta}(\rho_-))  -\frac{1}{4} \rho_+\rho_-\cos(\theta^+-\theta^-)) - \frac{\epsilon}{2}(\rho^+\sin(\theta^+) - \rho_-\sin(\theta^-))
\end{multline}
where 
\[
g_{\beta}(\rho)= -\frac{1}{8}\rho^2 - \frac{1}{2\beta} (S(\rho\hate_1).
\]
Considering the angular gradient, at any stationary point we have that 
\begin{align*}
&0= \nabla_{\theta^+} \phi= \frac{1}{4} \rho_+\rho_-\sin(\theta^+-\theta^-) - \frac{\epsilon}{2}\rho^+\cos(\theta^+)\\
&0= \nabla_{\theta^-} \phi= -\frac{1}{4} \rho_+\rho_-\sin(\theta^+-\theta^-)+ \frac{\epsilon}{2}\rho^-\cos(\theta^-)
\end{align*}
From this we conclude $\rho_-\cos(\theta^-)=\rho_+\cos(\theta^+)$.

Assume for the moment that  the variables $\rho_-\cos(\theta^-),\rho_+\cos(\theta^+) \neq 0$ (we will rule the other possibilities out \textit{as minimizers} below provided $\epsilon$ is small and $\beta$ is large).  Working under this assumption we show that $\rho_+$ must equal $\rho_-$.  The sum of angles formula for $\sin$ gives
\[\rho_+\sin(\theta^+)-\rho_-\sin(\theta^-) =2 \epsilon.
\]
as well.

Next radial differentiation gives
\[
 \nabla_{\rho_{\pm}} \phi= \frac{1}{2} \nabla g - \frac{\rho_{\mp}}{4} \cos(\theta^+-  \theta^-) \pm \frac{\epsilon}{2} \sin(\theta^{\pm})=0
\]
at stationary points.
Combining with the information gained from the angular differentiation, we find
\[
\rho_{\mp} \cos(\theta^+-  \theta^-) \pm \frac{\epsilon}{2} \sin(\theta^{\pm})=  \frac{1}{4}\rho_\pm.
\]
This implies that at stationary points both $\rho_+, \rho_-$ satisfy the equation
\[
 \nabla g(\rho) -\frac{1}{2}\rho =0.
\]
This gives the Mean Field Equation for the standard $XY$ model at inverse temperature $\beta $ once all factors of $\frac 12$ have been accounted for.

We introduce the notation $\Psi(m) = \psi_{\beta}(\|m\|_2)=   g(\|m\|_2) -\frac{1}{4}\|m\|_2^2=   - \frac 12 \|m\|_2^2 - \frac 1\beta S(\|m\|_2\hate_1)$ and record for reference (see \cite{KS}):
\begin{proposition}
\label{P:KS}
Let $\beta_c=2$.  For $\beta\leq  \beta_c$
\[
\partial_{\rho} \psi_{\beta}(\rho)= 0
\]
has the unique solution $\rho =0$ where as for $\beta> \beta_c$, there are three solutions $\{0, \rho_{\beta}, - \rho_{\beta}\}$.
Further, for $\beta \leq \beta_c$, $\rho = 0$ minimizes $\psi$ while for $\beta> \beta_c$ $\rho= 0$ is a local maximum and $\rho \in \{ \rho_\beta, - \rho_{\beta}\}$ are the unique global minima.  These are all the solutions to the mean field equation
\[
\rho =  R(\mathbf {\rho})= \frac{\int_{-1}^1  \frac{u}{\sqrt{1-u^2}} e^{\beta u\rho}\textd u} {\int_{-1}^1 (1-u^2)^{-\frac 12} e^{\beta u \mathbf {\rho}} \textd {u}}.
\]
\end{proposition}

We add to this proposition the following simple observations:
There exists $\beta_0>0$ and $\delta >0$ so that so that if $\beta> \beta_0> \beta_c$ and $\rho_i \in \{\rho: |\rho\pm \rho_\beta|< \delta\}$ then
\[
|R(\rho_1) -R(\rho_2)| \leq \beta e^{-c \beta}| \rho_1- \rho_2|
\]
since $\partial_{\rho} R= \beta \Var_{\beta \rho}(U)$.
Further, for $\rho \notin \{\rho: |\rho\pm \rho_\beta|< \delta\}$,
\[
\psi_{\beta}(\rho)- \psi_{\beta}(\mathbf{\rho}) \geq \frac 14,
\]
and finally $ |\rho_\beta| \geq \frac{1}{2}$.
Assume from now on that $\beta>\beta_0$.

Summarizing, for all $\epsilon$ sufficiently small ($\epsilon< \frac 1{16}$ is sufficient), $\beta > \beta_0$ and if $\rho_-\cos(\theta^-)=\rho_+\cos(\theta^+) \neq 0$, possible stationary points satisfy $\rho^{\pm} \in \{0, \mathbf \rho_{\beta}\}$.  By assumption on $(\beta, \epsilon)$, local minimizers must satisfy $\rho_{\pm} = \mathbf \rho_{\beta}$.  We then conclude from $\rho_-\cos(\theta^-)=\rho_+\cos(\theta^+)$ that $\theta_+=\pm \theta_- \neq 0$.  From $\rho_+\sin(\theta^+)-\rho_-\sin(\theta^-) =2 \epsilon$ it must be that $\theta^-= - \theta^+= - \theta$ and  $\sin(\theta) =\frac{\epsilon}{\rho}$.  Let us call this solution $(\mathbf m^+, \mathbf m^-)$.

We must still rule out the cases $\rho_-\cos(\theta^-)=\rho_+\cos(\theta^+) = 0$.  It is easy to reduce to the two scenarios $\rho_-=\rho_+=0$ or $\theta^-=\theta^+ \in \{\frac \pi2, \frac{3\pi}{2}\}$ as the other possibilities lead to larger free energies than these two.  Both cases fall within the optimization problem
\[
\min_{\rho} \psi_{\beta}(\rho).
\]
Since $\beta>\beta_0$, we see that of the two, $(m^+_0, m^-_0)= \rho_{\beta}(\hate_2, \hate_2)$ has the lower free energy.  Then the free energy difference of the latter from the absolute minimum is difference
\[
\phi(\mathbf m^+, \mathbf m^-)- \phi(m^+_0, m^-_0) = - \frac{\epsilon^2}{2}
\]
by direct computation.

Finally we consider stability of the optimizers.  Again, we assume $\beta> \beta_0$.  Suppose that $\xi \leq \delta \wedge \frac{\epsilon}{4}$ and let 
\[
A_{\xi}=\{(m_+, m_-): \min_{\pm} \|(m_+, m_-) \pm (\mathbf m^+, \mathbf m^-)\|\leq \xi\}.
\]  
We are interested in $\phi(m^+_0, m^-_0)- \phi(\mathbf m^+, \mathbf m^-)$ when $(m^+_0, m^-_0) \notin A_{\xi}$.  Because we have identified all stationary points of $ \phi$, we have that
\[
  \min_{(m^+_0, m^-_0) \in A^c_{\xi}} \phi(m^+_0, m^-_0)-\phi(\mathbf m^+, \mathbf m^-) = \min_{(m^+_0, m^-_0) \in \partial A_{\xi}} \phi(m^+_0, m^-_0)-\phi(\mathbf m^+, \mathbf m^-) \wedge \frac{\epsilon^2}{2}.
\]
Thus we only need to compute a lower bound on the free energy difference on $\partial A_{\xi}$.  

Let $R_X$ denote the reflection across the $\hate_1$ axis in $\R^2$. Then
\[
\phi(m^+_0, m^-_0) - \frac{1}{2}\left\{  \phi(m^+_0, R_X m^+_0) +  \phi(R_X m^-_0, m^-_0)\right\}
= \frac{1}{8}( m^+_0- R_X m^-_0 \cdot R_X m^+_0- m^-_0).
\]
This last term is positive if $(m^+_0, m^-_0) \in A_{\xi}$ and $\epsilon$ is sufficiently small so that the Taylor expansion of cosine about $0$ is accurate ($\epsilon< \frac 1{16}$ will do fine).  
Thus to obtain a lower bound over $\partial A_{\xi}$, we may assume $m^-_0= R_X m^+_0$.
Under this condition, the functional simplifies to 
\[
\phi(m^+_0, R_X m^+_0) =\chi(\rho, \theta):= -\frac{\rho^2}{2} - \frac{1}{\beta} S(\rho)
 + \frac{\rho^2 \sin^2 \theta}{2} -\epsilon \rho \sin \theta \]
where $\rho= \|m^+_0\|_2$ and $m^+_0 \cdot R_X m^+_0 = \cos(2\theta)$.  
We may compute 
\[
\nabla^2 \chi(\rho, \theta)=\left(\begin{array}{cc}-1 - \frac{1}{\beta} \partial^2_{\rho} S(\rho) + 2 \sin^2 \theta & \rho\sin(2 \theta) - \epsilon \cos(\theta) \\\rho\sin(2 \theta) - \epsilon \cos(\theta) & \rho^2\cos(2\theta) + \epsilon \rho \sin \theta\end{array}\right).
\]
The on-diagonal terms are both $O(1)$ while the off diagonal terms are $O(\epsilon)$ for $\beta> \beta_0$ for $(m^+_0, R_X m^+_0) \in A_{\xi}$.
This implies
\[
\phi(m^+_0, R_X m^+_0) - \phi(\mathbf m^+, \mathbf m^-) \geq C \{(\|m^+_0\|_2-\rho_{\beta})^2 +  (\theta-\theta_{\beta})^2\}.
\]
Since Euclidean length in polar coordinates and Cartesian coordinates are equivalent in $A_{\xi}$ for $\epsilon$ sufficiently small and $\beta> \beta_0$, the proposition is proved.
\end{proof}

For $ m$ in $\{x: \|x\|_2< 1\}$ the mapping $m \mapsto -\nabla S(m)$ has an inverse, defined on $\R^2$.  This inverse is given by
\[
M(h)= \frac{\int  \sigma e^{\sigma \cdot h} \textrm{d}\nu(\sigma) }{\int  e^{\sigma \cdot h} \textrm{d}\nu(\sigma)}
\]

Let
\[
M^*(h)= M_{\beta, \epsilon}^*(h) = \frac{1}{2}\left(M(\beta(h + \epsilon \hate_2)) + M(\beta(h - \epsilon \hate_2))\right)
\]
Note that stationary points of $\phi$ satisfy
\[
\bar m = M^*(\bar m).
\]
For later reference, we will need estimates on the difference
\[
M^*(h)-M^*(\bar{\mathbf m})
\]
subject the condition $\|h- \bar{\mathbf m}\|_2 < \delta$.

We have
\begin{proposition}
\label{P:Contract}
There exists $\delta, \epsilon_0, \beta_0>0$ so that if $\epsilon< \epsilon_0$ and $\beta> \beta_0$ and if 
\[
h \in \{ \|h'- \bar{\mathbf m}\|_2 < \delta\},
\]
then we have
\[
\|M^*(h_1)-M^*(\bar{\mathbf m})\|_2 \leq  \left(4\beta e^{-c_1 \beta} + 1 - c_2 \epsilon^2\right) \|h-\bar{\mathbf{m}}\|_2 +O(\|h_1-\bar{\mathbf{m}}\|_2^2)
\]
where $c_1, c_2$ depend only on $\epsilon_0, \beta_0, \delta$.
\end{proposition}
\begin{proof}
Provisionally, let $\delta, \beta_0, \epsilon_0$ be as described after Proposition \ref{P:KS}.  To see the stated contraction occurs, we consider two possibilities.

\noindent
\textbf{Case 1: $\|m\|_2 \geq \|h\|_2$.}
Note that
\begin{equation}
\label{E:M*}
M(\beta h)= R(\|h\|_2) \hat{h}.
\end{equation}
Then
\[
M^*(h)= \frac 12\left( R(\|h + \epsilon \hate_2\|) \widehat{h+ \epsilon \hate_2} + R(\|h - \epsilon \hate_2\|_2) \widehat{h-\epsilon \hate_2}\right).
\]
Now by hypothesis on $\delta$
\[
|R(\|h \pm \epsilon \hate_2\|_2)- R(\|\bar{\mathbf m} \pm \epsilon \hate_2\|_2)| < \beta e^{-c \beta} \| h- \bar{\mathbf m}\|_2.
\]
Since
\[
R(\|\bar{\mathbf m} + \epsilon \hate_2\|_2)= R(\|\bar{\mathbf m} - \epsilon \hate_2\|_2)
\]
the Proposition is proved by estimating
\begin{equation}
\label{E:Pert}
\frac{\rho_{\beta}^2}{4}\| \widehat{h+ \epsilon \hate_2} + \widehat{h-\epsilon \hate_2} -  \widehat{\bar{\mathbf m}+ \epsilon \hate_2} - \widehat{\bar{\mathbf m}-\epsilon \hate_2}\|^2_2.
\end{equation}
Letting $\hat{h}$ denote the unit vector in the direction of $h$ and write $\hat h= a \hate_1 + b \hate_2$, we may rewrite
\begin{multline}
\| \widehat{h+ \epsilon \hate_2} + \widehat{h-\epsilon \hate_2} -  \widehat{\bar{\mathbf m}+ \epsilon \hate_2} - \widehat{\bar{\mathbf m}-\epsilon \hate_2}\|^2_2\\
= \| \widehat{\hat{h}+ \frac{\epsilon}{\|h\|_2} \hate_2} + \widehat{\hat{h}-\frac{\epsilon}{\|h\|_2} \hate_2} -  2 \frac{\|\bar m\|}{\rho_\beta} \hate_1\|^2_2.
\end{multline}

Let $\mu= \epsilon^2( \|\bar{\mathbf m}\|_2- \|h\|_2)$.  After a tedious perturbation theory calculation we have \eqref{E:Pert} bounded by
\begin{equation}
\label{E:P2}
\|\bar{\mathbf m}\|_2^2 \left[(1-a - c_1\mu)^2+ b^2(1-c_2 \epsilon^2)^2\right] + O( \|\bar{\mathbf m} - h\|^3_2)
\end{equation}
for two universal constants $c_1, c_2>0$.
Now $\mu=O(\epsilon^2  \|\bar{\mathbf m} - h\|_2)$.  Up to terms of order  $O( \|\bar{\mathbf m} - h\|^3_2)$, $\|\bar{\mathbf m}\|_2^2((1-a)^2+ b^2) = \|\bar{\mathbf m} - h\|^2_2$ so multiplying and dividing by $(1-a)^2+ b^2$ we have \eqref{E:Pert} bounded by
\[
\frac{(1-a - c_1\mu)^2+ b^2(1-c_2 \epsilon^2)^2}{(1-a)^2+ b^2} \|\bar{\mathbf m} - h\|^2_2 +  O( \|\bar{\mathbf m} - h\|^3_2).
\]
By assumption on the Case, $\mu>0$.
The required factor $1- c_3 \epsilon^2$ now follows whenever $\|\bar{\mathbf m} - h\|_2< \frac 1 {16}$ since $a<1$ and $\| \bar{\mathbf m} \|_2 \geq \frac 12$ for $\beta>\beta_0$ and $\epsilon< \epsilon_0$.

\textbf{Case 2: $\|\bar{\mathbf m}\|_2 \leq \|h\|_2$.}
Let $h_{\beta}= \rho_{\beta} \hat h$.
For $\beta>\beta_0$ and $\|h-\bar{\mathbf{m}}\|_2< \delta$, it is always true that
\[
|R(\|h \pm \epsilon \hate_2\|_2)- R(\| h_{\beta} \pm \epsilon \hate_2\|_2)| < \beta e^{-c \beta} \| h- \bar{\mathbf m}\|_2.
\]
Since $\|h_{\beta}- \bar{\mathbf m}\|_2 \leq \|h- \bar{\mathbf m}\|_2$ for $h \in \{x: \|x\|_2= \|\bar{\mathbf{m}}\|_2\}$ this case reduces to Case 1.
\end{proof}

\textbf{Entropy Estimates:}
We will need quantitative estimates on the entropy function $S(m)$ as well as its finite volume approximations.
Let
\[
\AA_{\delta, N}(m):= \left \{ \sigma \in \mathbb S_1^{[N]} \| \frac{1}{N} \sum_{i=1}^N \sigma_i - m\|_{2} < \delta\right \},
\]
which is a subset of $\mathcal S_N:=(\mathbb S_1)^{[N]}$ and let
\[
S_{\delta, N}(m):= \frac{1}{N} \log \nu_N(\AA_{\delta, N}(m))
\]
where  $\nu_N$ is the independent product measure on $\mathcal S_N$ spins with each one coordinate marginal given by Haar measure on $\mathbb S^1$.
Both $S(m)$ and $S_{\delta, N}(m)$ are rotationally invariant so we may assume $m= \rho \hate_1$ and consider them as functions of $\rho$, call them $S(\rho)$ and $S_{\delta, N}(\rho)$.

For each $\rho \in (-1,1)$ let $h(\rho)= \textrm{argmin}_h (G(h \hate_1)-  \rho h)$.  Convexity implies that $h(\rho)$ exists and is unique:
\begin{proposition}
For each $\rho \in (-1, 1)$, the equation $h(\rho)= \textrm{argmin}_h (G(h \hate_1)-  \rho h)$ has a unique solution which satisfies the equation
\[
\rho= \frac{\int_{-1}^1 u (1- u^2)^{-\frac 12}e^{u h} \textrm{d} u}{\int_{-1}^1 (1- u^2)^{-\frac 12}e^{u h} \textrm{d} u}.
\]
In particular, $h(\rho)$ diverges as
\[
c \leq \frac{|h(\rho)|}{-\log(1-|\rho|)} \leq C
\]
as $|\rho| \uparrow 1$.
\end{proposition}
Below we will need quantitative bounds on the deviation of $S_{\delta, N}(m)$ from the entropy function $S(m)$.  The following will be sufficient.
\begin{lemma}
\label{L:Ent}
Suppose that $\delta>0$. Then we have
\[
|S(\rho)- S_{\delta, N}(\rho)| \leq  C \delta |h(\rho)|\vee \delta^{-2} N^{-2}.
\]
\end{lemma}
\begin{proof}
By definition of $h (\rho)$
\[
e^{-S(\rho)N+ S_{\delta, N}(\rho)N} = e^{O(\delta|h(\rho)| N)} \mu_{N, h(\rho)}(A_{\delta, N}(\rho \hate_1))
\]
with $\mu_{N, h(\rho)}$ the $N$ spin probability measure tilted by $h(\rho) \hate_1$.  Because $\mu_{N, h(\rho)}$ is a probability measure, we only need to provide a lower bound on
\[
 \mu_{N, \rho(h)}(A_{\delta, N}(\rho  \hate_1)).
\]
The claimed bound follows from a simple application of Chebyshev's Inequality applied to $A_{\delta, N}^c$.
\end{proof}

\section{Free Energy Functional Estimates for Clean Contours}
By definition, if $\sigma \in \mathbb X(\Gamma)$, the restriction of $\Theta_z(\sigma)$ to each of the components $\delta_{ext}(\Gamma), \delta^{i}_{in}(\Gamma)$ is constant.
We shall say that $\Gamma$ is a $\pm$ contour if $\Theta_z(\Gamma)=\pm 1$ on $\delta_{ext}(\Gamma)$ (this notion makes sense by specification that $\Gamma$ is a contour).  Denoting $\delta^{\pm}_{ext}(\Gamma)= \delta_{ext}(\Gamma)$ in case of a $\pm$ contour and the empty set otherwise we let $R^{\pm}_i(\Gamma)= R^{\pm} \cap \left(\delta^{i}_{in}(\Gamma)\right)$ and
$\delta^{\pm}_{ext}(\Gamma) \cup (\cup_i R^{\pm}_i(\Gamma))= R^{\pm} (\Gamma)$.

Given that $\Gamma$ is a clean $+$ contour, we say that a boundary spin configuration $\sigma_{\delta(\Gamma)^c}$ is compatible with $\Gamma$ if
\[
\mu_{\delta(\Gamma)}^{\sigma_{\delta(\Gamma)^c}}\left(\theta_z(\sigma'_{\delta(\Gamma)})= \theta_{\Gamma}(z) \text{ for } z \in \Sp(\Gamma) \text{ and } \Theta_z(\sigma'_{\delta(\Gamma)})= \pm 1\text{ for } z \in R^{\pm}(\Gamma)\right)\neq 0.
\]
Note here that $ \Theta_z(\sigma'_{\delta(\Gamma)})$ implicitly takes as an argument the extended configuration $(\sigma'_{\delta(\Gamma)}, \sigma_{\delta(\Gamma)^c})$ although we will continue to suppress this detail below.

For any such $\sigma_{\delta(\Gamma)^c}$, let
\[
W(\Gamma; \sigma_{\Lambda^c})=
 \frac{\mu_{\delta(\Gamma)}^{\sigma_{\delta(\Gamma)^c}}\left(\theta_z(\sigma'_{\delta(\Gamma)})= \theta_{\Gamma}(z) \text{ for } z \in \Sp(\Gamma) \text{ and } \Theta_z(\sigma'_{\delta(\Gamma)})= \pm 1\text{ for } z \in R^{\pm}(\Gamma)\right)}{\mu_{\delta(\Gamma)}^{\Theta \times \sigma_{\delta(\Gamma)^c}}\left( \Theta_z(\sigma'_{\delta(\Gamma)})= 1\text{ for } z \in \delta(\Gamma)\right)}
\]
where $ \Theta \times \sigma_{\delta(\Gamma)^c}$ denotes the boundary condition with $(\Theta \times \sigma_{\delta(\Gamma)^c})_z\equiv  \Theta_z(\Gamma) \sigma_{\delta(\Gamma)^c, z}$.
Let
\[
\|W(\Gamma; \cdot)\| = \sup_{\{\sigma_{\Lambda^c} \text{ compatible}\}} W(\Gamma; \sigma_{\Lambda^c})
\]
where $\{\sigma_{\delta(\Gamma)^c} \text{ compatible}\}$ indicates that we only consider boundary conditions for which the numerator does not vanish.
Notions for $-$ contours are defined similarly with $+$ and $-$ reversed and, in particular, $ \Theta \times \sigma_{\Lambda^c}$ is replaced by $- \Theta \times \sigma_{\Lambda^c}$.

The main result which allows us to proceed is:
\begin{lemma}
\label{L:Key}
There exist $\delta, \epsilon_0, \beta_0>0$ so that if $\epsilon< \epsilon_0$, $\beta> \beta_{\epsilon}> \beta_0$, $0< \xi< \delta$ and $p, \lambda, \kappa \in (0, \frac 13)$ then the following holds:  Suppose that $\Gamma \subset \Z^d$ is a $(\kappa, p)$-clean contour with respect to $\omega \in \Omega$.  Then
\[
\|W(\Gamma; \cdot)\| \leq e^{-q_1 N_{\Gamma}}
\]
where
\[
q_1 = C_1\left\{- \beta(\epsilon^2 \wedge \xi^{2} \wedge L^{-\lambda d}) L^{(1-\lambda)d} +  C(\beta (  L^{\lambda-1}+ \epsilon p + \epsilon L^{(-d/2+d\kappa)(1- \lambda)} + e^{- c \epsilon^2 L^{\lambda}})+1)L^{(1+\lambda)d} \right\}.
\]
\end{lemma}

Let us attend to the proof of Lemma \ref{L:Contours} before exposing the proof of Lemma \ref{L:Key}.
\begin{proof}[Proof of Lemma \ref{L:Contours}]
Let $N$ be fixed and consider the event $\mathbb X(\Gamma_1^*, \dotsc, \Gamma_m^*, \Gamma_{m+1},\dotsc, \Gamma_{m+n})$ with $\Sp(\Gamma_i) \subset \Lambda_N$.  Then we claim
\[
\mu_{\Lambda_N}^{\ra}(\mathbb X(\Gamma_1^*, \dotsc, \Gamma_m^*, \Gamma_{m+1},\dotsc, \Gamma_{m+n})) \leq \prod_{i=1}^m\|W(\Gamma_i; \cdot)\|
\]
Once this is justified, the Lemma is proved by application of Lemma \ref{L:Key}.

The proof of this claim proceeds by induction on $m$.  Interpreting an empty product as $1$, the case $m= 0$ there is nothing to prove, so we proceed to the induction step. Suppose the claim is true for any $*$ compatible system with $m=k$ clean contours and $n$ dirty contours.  Given a set $\{\Gamma_1^*, \dotsc, \Gamma_{k+1}^*, \Gamma_{k+2},\dotsc, \Gamma_{k+n+1})\}$ of $*$-compatible contours and reordering as necessary, we may assume
\[
c(\Gamma_{k+1}) \cap \cup_{i=1}^k \Sp\{\Gamma_i\} = \varnothing.
\]

Assume for concreteness that $\Gamma_{k+1}$ is a $+$ contour.  The argument in the case of a $-$ contour proceeds in a similar manner.
Let $\{\tilde \Gamma_1, \dotsc, \tilde \Gamma_r\}$ denote the set of contours among $\{\Gamma_{m+1},\dotsc, \Gamma_{m+n}\}$ with $\delta(\tilde{\Gamma}_i) \subset \Int(\Gamma_{k+1})$ and $\{\tilde \Gamma'_1, \dotsc, \Gamma'_{n-r}\}$ denote the rest.  Then we have
\[
\mathbb X(\Gamma_1^*, \dotsc, \Gamma_{k+1}^*, \Gamma_{k+2},\dotsc, \Gamma_{k+n+1})= \mathbb X(\Gamma_1^*, \dotsc, \Gamma_k^*, \Gamma'_{1},\dotsc, \Gamma'_{n-r}) \cap \mathbb X(\Gamma_{k+1}^*, \tilde \Gamma_1, \dotsc, \tilde \Gamma_r)
\]
Using the DLR equations we have
\begin{multline}
\mu_{\Lambda_N}^{\ra}(\mathbb X(\Gamma_1^*, \dotsc, \Gamma_{k+1}^*, \Gamma_{k+2},\dotsc, \Gamma_{k+n+1}))\\
=\left\langle \mathbf 1_{\mathbb X(\Gamma_1^*, \dotsc, \Gamma_k^*, \Gamma'_{1},\dotsc, \Gamma'_{n-r})} \mathbf 1_{\{\theta|_{\delta^{=}_{ext}(\Gamma_{k+1})}(\sigma) \equiv 1\}} \langle \mathbf 1_{\mathbb X(\Gamma_{k+1}^*,\tilde{\Gamma}_{1},\dotsc, \tilde{\Gamma}_{r})}\rangle_{\Sp(\Gamma_{k+1})\cup \Int(\Gamma_{k+1})}^{\sigma_{\delta_{ext}(\Gamma_{k+1})}} \right \rangle_{{N}}^\ra
\end{multline}
where $\{\theta|_{\delta^{=}_{ext}(\Gamma_{k+1})}(\sigma) \equiv 1\}$ indicates that that  the phase function $\theta$ is one on $\delta_{ext}(\Gamma_{k+1})$ \textit{and the exterior $\ell_{>}$-measurable blocks neighboring $\delta_{ext}(\Gamma_{k+1})$.}

If $\Gamma= (\Sp(\Gamma), \theta_{\Gamma})$, let $-\Gamma= (\Sp(\Gamma), -\theta_{\Gamma})$ and let 
\[
T_{\Gamma_{k+1}}(\tilde{\Gamma}_{\ell})= \begin{cases}
\tilde{\Gamma}_{\ell} \text{ if $\delta_{ext}(\tilde{\Gamma}_{\ell}) \subset \delta^+_{in}(\Gamma_{k+1}))$}\\
-\tilde{\Gamma}_{\ell} \text{ otherwise}.
\end{cases}
\]
Note that this transformation preserves the notion of cleanliness.
It is straightforward, using the invariance of internal energy under spin reflections about the $Y$-axis, to see that
\begin{multline}
1_{\{\theta|_{\delta^=_{ext}(\Gamma_{k+1})}(\sigma) \equiv 1\}} \langle \mathbf 1_{\mathbb X(\Gamma_{k+1}^*,\tilde{\Gamma}_{1},\dotsc, \tilde{\Gamma}_{r})}\rangle_{\Sp(\Gamma_{k+1})\cup \Int(\Gamma_{k+1})}^{\sigma_{\delta_{ext}(\Gamma_{k+1})}}\\ \leq \|W(\Gamma_{k+1}; \cdot)\|
1_{\{\theta|_{\delta^{=}_{ext}(\Gamma_{k+1})}(\sigma) \equiv 1\}} \langle \mathbf 1_{\mathbb X(T_{\Gamma_{k+1}}(\tilde{\Gamma}_{1}),\dotsc, T_{\Gamma_{k+1}}(\tilde{\Gamma}_{r}))}\rangle_{\Sp(\Gamma_{k+1})\cup \Int(\Gamma_{k+1})}^{\sigma_{\delta_{ext}(\Gamma_{k+1})}}.
\end{multline}
Since
\begin{multline}
\left\langle \mathbf 1_{\mathbb X(\Gamma_1^*, \dotsc, \Gamma_k^*, \Gamma'_{1},\dotsc, \Gamma'_{n-r})} \mathbf 1_{\{\theta|_{\delta^{=}_{ext}(\Gamma_{k+1})}(\sigma) \equiv 1\}} \langle \mathbf 1_{\mathbb X(T_{\Gamma_{k+1}}(\tilde{\Gamma}_{1}),\dotsc, T_{\Gamma_{k+1}}(\tilde{\Gamma}_{r}))}\rangle_{\Sp(\Gamma_{k+1})\cup \Int(\Gamma_{k+1})}^{\sigma_{\delta_{ext}(\Gamma_{k+1})}} \right\rangle_{N}^\ra\\
\leq \mu_{\Lambda_N}^{\ra}(\mathbb X(\Gamma_1^*, \dotsc, \Gamma_{k}^*, \Gamma'_{1},\dotsc, \Gamma'_{n-r},T_{\Gamma_{k+1}}(\tilde{\Gamma}_{1}),\dotsc, T_{\Gamma_{k+1}}(\tilde{\Gamma}_{r})))
\end{multline}
the induction step is proved.
\end{proof}

\noindent
\textbf{Proof of Lemma \ref{L:Key}:Reduction to Deterministic Weights}

\noindent
For approximation purposes we introduce a block mean field Hamiltonian.  In the formula $B_r$ denote the $\ell_<$-measurable blocks in $\Lambda$, $B_r^{\pm}$ denotes any fixed  partition of $B_r$ into sets with cardinalities differing by at most one (we will call this an \textit{equal splitting} below), and we let
\[
J^{\ell_<}_L(z, z')\equiv J_{L}(r, r')
\]
if $z \in B_r, \: z' \in B_r'$ and $B_r, B_r'$ are both $\ell_<$-measurable, so that $J^{\ell_<}_L(z, z')$ is constant over pairs of $\ell_<$-measurable blocks.  Let us define
\begin{multline}
\tilde{\scrH_\Lambda}(\sigma_\Lambda|\sigma^{\ell_<}_{\Lambda^c}) = -  \frac{1}{2} \sum_{z, {z'} \in \Lambda} J^{\ell_<}_L(z, z') \:  \sigma^{\ell_<}_z \cdot \sigma^{\ell_<}_{z'}  - \frac{1}{2}\sum_{z\in \Lambda, z' \in \Lambda^c}  J^{\ell_<}_L(z, z') \sigma^{\ell_<}_z \cdot \sigma^{\ell_<}_{z'}\\+
  \sum_{r \in \Lambda \cap \{(2\ell_<+1)\Z\}^d} \left[\sum_{z \in B_r^+}- \frac{1}{2} \epsilon \splus_z\cdot \hate_2 +\sum_{z \in B_r^-} \frac{1}{2}  \epsilon \sminus_z\cdot \hate_2 \right]. 
\end{multline}
Let $\tilde \mu^{\sigma^{\ell_{<}}}_{\Lambda}$ denote the corresponding Gibbs measures.  We introduce contour weights by
\begin{equation}
\label{E:Det}
\tilde{W}(\Gamma; \sigma_{\delta(\Gamma)^c})=
 \frac{\tilde \mu_{\delta(\Gamma)}^{\sigma_{\delta(\Gamma)^c}}\left(\theta_z(\sigma_{\delta(\Gamma)})= \theta_{\Gamma}(z) \text{ for } z \in \Sp(\Gamma) \text{ and } \Theta_z(\sigma'_{\delta(\Gamma)})= \pm 1\text{ for } z \in R^{\pm}(\Gamma)\right)}{\tilde \mu_{\delta(\Gamma)}^{\Theta \times \sigma_{\delta(\Gamma) ^c}}\left(\Theta_z(\sigma'_{\delta(\Gamma)})= 1\text{ for } z \in \delta(\Gamma)\right)}
 \end{equation}
defined whenever $\sigma_{\delta(\Gamma)^c}$ is compatible with $\Gamma$.
\begin{lemma}
\label{L:Aprox}
Suppose that $\Gamma$ is a fixed contour.  For any equal splitting of $\ell_<$-measurable boxes $B_r$ into $B_r^{\pm}$ we have
\[
\left|\log \frac{\tilde{W}(\Gamma; \sigma_{\delta(\Gamma)^c})}{W(\Gamma; \sigma_{\delta(\Gamma)^c})}\right| \leq q_2(p, \lambda, \kappa) N_{\Gamma}
\]
uniformly over the set of boundary conditions $\sigma_{\delta(\Gamma)^c}$ compatible with $\Gamma$,
where
\[
q_2(p, \lambda, \kappa) = C \beta ( L^{\lambda-1}+ \epsilon p + \epsilon L^{(-d/2+d\kappa)(1- \lambda)})\ell^d_>
\]
\end{lemma}

\begin{proof}[Proof of Lemma \ref{L:Aprox}]
The key point here is that up to errors of order
\[
c\beta (L^{\lambda-1}+ \epsilon p + \epsilon L^{(-d/2+d\kappa)(1- \lambda)})\ell^d_>
\]
the precise location of the local fields may be forgotten up to the information that most $\ell_<$ measurable boxes have well balanced local field statistics.

It is easy to see that $\tilde W$ depends on the precise choices of $B^{\pm}_r$ only up to $O(\beta L^{\lambda-1} |\Lambda|)$.  Thus we may work with a suitable choice of equal splitting.  The choice we make is as follows: for each $B_r$ $\ell_<$-measurable, we choose $B^{\pm}_r$ so that whenever $N^{\pm}_r$ has smaller cardinality than $N^{\mp}_r$, then $N^{\pm}_r\subset B^{\pm}_r$ and then the rest of $B^{\pm}_r$ is filled out by an arbitrary subset of $N^{\mp}_r$.

With this choice, we claim
\[
\left|\scrH(\sigma_{\Lambda}| \sigma_{\Lambda^c}) - \tilde{\scrH}(\sigma_{\Lambda}| \sigma_{\Lambda^c})\right|\leq    C(L^{\lambda-1} +\epsilon L^{(1- \lambda)(-d/2 +\kappa d)} + \epsilon p) |\Lambda|.
\]
This leads immediately to the result.  For future reference, we record this statement outside the proof along with some other internal energy approximations.
\end{proof}

For the particular equal splitting of $\ell_<$-measurable boxes $B_r$ introduced in the previous proof
let
\[
\sigma^{*, \ell_<, +}_{\Lambda, z}, \sigma^{*, \ell_<, -}_{\Lambda, z}
\]
denote the block averages:
\[
\sigma^{*, \ell_<, \pm}_{\Lambda, z}= \frac{2}{|B_r|} \sum_{x \in B^{\pm}_r} \sigma_x.
\]
Below, it is implicitly understood that these block averages may be consider as functions on $L^{\infty} \times L^{\infty}$ by extending them to be $\ell<$-piecewise constant over $\hat{\delta(\Gamma)}$.
Let $(m^+, m^-)$ be a two component function with each component in $L^{\infty}(\Lambda)$ and let  the boundary condition $\bar m_{\Lambda^c} \in L^{\infty}(\Lambda^c)$.  Introduce the continuum internal energy
\begin{multline}
U^{(L)}(m^{+}_{\Lambda}, m^{-}_{\Lambda}| \bar m_{\Lambda^c})=-  \frac{1}{2} \int_{\Lambda \times \Lambda} \textd z \textd z' J_L(z, z') \bar m_z \cdot \bar{m}_{z'} \\-
 \frac{1}{2} \epsilon \int_{\Lambda}\textd z  m^+_z\cdot \hate_2 +\frac{1}{2}  \epsilon \int_{\Lambda}m^-_z\cdot \hate_2
   -\frac{1}{2} \int_{\Lambda, \Lambda^c} \textd z \textd z' J_L(z, z')\bar m_z \cdot \bar m_{\Lambda^c, z}
\end{multline}

\begin{proposition}
\label{P:EAprox}
Let us suppose that $\Lambda$ is a $(\kappa, p)$- clean $\ell_>$-measurable subset for $\omega \in \Omega$.
Then
\[
\left|\scrH(\sigma_{\Lambda}| \sigma_{\Lambda^c}) - \tilde \scrH(\sigma_{\Lambda}| \sigma_{\Lambda^c})\right|\leq    c (L^{\lambda-1} +\epsilon L^{(1- \lambda)(-d/2 +\kappa d)} +\epsilon p) |\Lambda|
\]
with respect to the splitting $B^{\pm}_r$ introduced above.
For any equal splitting,
\[
\left|\tilde \scrH(\sigma_{\Lambda}| \sigma_{\Lambda^c}) - U^{(L)}(\sigma^{*, \ell_<, +}_{\Lambda}, \sigma^{*, \ell_<, -}_{\Lambda} | \sigma_{\Lambda^c})\right|\leq    c ( L^{\lambda-1} +\epsilon L^{(1- \lambda)(-d/2 +\kappa d)} + \epsilon p) |\Lambda|.
\]
Further, if $(m^+, m^-)$ is an arbitrary magnetization profile in $L^{\infty}(\R^d)\times L^{\infty}(\R^d)$ then
\[
\left|U^{(L)}(m^+_{\Lambda}, m^-_\Lambda| \bar m_{\Lambda^c})- U^{(L)}(m^{\ell_<, +}_{\Lambda}, m^{\ell_<, -}_{\Lambda}| m^{\ell_<}_{\Lambda^c})\right| \leq  c L^{\lambda-1} |\Lambda|.
\]
\end{proposition}

\begin{proof}
This is a relatively straightforward application of the various assumptions.
Let us begin with the claim for spin configurations:  Since $J_L$ is slowly varying, if $r, r'$ are the centers of two $\ell_<$-measurable boxes $B_r, B_r'$, then
\[
|J_L(z, z')-J_{L}(r, r')| \leq C \|\nabla J\|_{L^\infty} L^{\lambda-1}.
\]
Thus we may replace
the term quadratic in the spins by
\begin{equation}
\label{E:tiny}
-  \frac{1}{2} \sum_{z, {z'} \in \Lambda} J^{\ell_<}_L(z, z') \:  \sigma^{\ell_<}_z \cdot \sigma^{\ell_<}_{z'}
  - \frac{1}{2}\sum_{z\in \Lambda, z' \in \Lambda^c}  J^{\ell_<}_L(z, z')\sigma^{\ell_<}_z \cdot \sigma^{\ell_<}_{z'} + O(L^{\lambda-1})
\end{equation}
For any $\ell_<$-piecewise constant profile, the first term is also the integral
\[
-  \frac{1}{2} \int_{\Lambda \times \Lambda} \textd z \textd z' J^{\ell_<}_L(z, z')\sigma^{\ell_<}_z \cdot \sigma^{\ell_<}_{z'}  \\-
    -\frac{1}{2} \int_{\Lambda, \Lambda^c} \textd z \textd z' J^{\ell_<}_L(z, z')\sigma^{\ell_<}_z \cdot \sigma^{\ell_<}_{z'}
\]
and using the approximation $|J_L(z, z')-J_{L}(r, r')| \leq C \|\nabla J\|_{L^\infty} L^{\lambda-1}$ once again, we have
\begin{multline}
-  \frac{1}{2} \sum_{z, {z'} \in \Lambda} J_L(z, z') \: \sigma^{\ell_<}_z \cdot \sigma^{\ell_<}_{z'}
  - \frac{1}{2}\sum_{z\in \Lambda, z' \in \Lambda^c}  J_L(z, z') \sigma^{\ell_<}_z \cdot \sigma^{\ell_<}_{z'}=\\
  -  \frac{1}{2} \int_{\Lambda \times \Lambda} \textd z \textd z' J_L(z, z')\sigma^{\ell_<}_z \cdot \sigma^{\ell_<}_{z'}  -
    -\frac{ 1}{2} \int_{\Lambda, \Lambda^c} \textd z \textd z' J_L(z, z')\sigma^{\ell_<}_z \cdot \sigma^{\ell_<}_{z'} + O(L^{\lambda-1} |\Lambda|).
\end{multline}
The middle claim of the Proposition follows due to the fact that the block Hamiltonian is defined relative to an equal splitting which is coupled appropriately to the randomness.

For the first claim of the proposition, we still need to estimate the difference in energy contributed by the local field terms of $\scrH(\sigma_{\Lambda}| \sigma_{\Lambda^c}) - \tilde \scrH(\sigma_{\Lambda}| \sigma_{\Lambda^c})$, whenever $\Lambda$ is $(\kappa, p)$ clean, then $|\frac{|N^+_r|}{N}- 1/2| < \ell_<^{-d/2+\kappa d}$ except over a bad set of boxes with total number at most $\frac{p|\Lambda|}{\ell_<^d}$.
Therefore
\begin{multline}
  \sum_{r \in \Lambda} \left[\sum_{z \in N_r^+}- \frac{1}{2} \epsilon \splus_z\cdot \hate_2 +\sum_{z \in N_r^-} \frac{1}{2}  \epsilon \sminus_z\cdot \hate_2\right]  - \\
   \sum_{r \in \Lambda} \left[\sum_{z \in B_r^+}- \frac{1}{2} \epsilon \sigma^{*, \ell_<, +}_z\cdot \hate_2 +\sum_{z \in B_r^-} \frac{1}{2}  \epsilon \sigma^{*, \ell_<, -}_z\cdot \hate_2\right]
 = O(\epsilon \ell_<^{-d/2+\kappa d}|\Lambda|) + O(\epsilon p|\Lambda|).
  \end{multline}
  Combining this with \eqref{E:tiny} proves the first claim.
  Finally, the approximation for magnetization profiles follows a similar argument and will not be given.
\end{proof}

\noindent
\textbf{Proof of Lemma \ref{L:Key}: Deterministic Free Energy Estimates.}
According to Lemma \ref{L:Aprox}, to prove Lemma \ref{L:Key}, it is enough for us to work with $\tilde{W}$ for some arbitrary but fixed equal splitting of $\ell_<$-measurable boxes $B_r^{\pm}$.  From now on the notation $(\splus_z, \sminus_z)$ refers to spatial averages taken with respect to this equal splitting.

Let $\Lambda$ be a bounded $\ell_>$-measurable region in $\Z^d$.  Let $\mathcal B^{\ell_<}_{0, \Lambda}$ denote the sub-sigma algebra of $\mathcal B_{0, \Lambda}$ with events determined by the block average profiles $\sigma^{\ell_<, \pm}_z$.  
An event $\AA$ in $\mathcal B^{\ell_<}_{0, \Lambda}$ may be identified with a subset $\AA^*$ of deterministic magnetization profiles on $\hat{\Lambda}$ in an obvious way:  for $(m^+_z, m^-_z)_{z \in \hat{\Lambda}}$ with $\|m^{\pm}\|_{L^\infty}\leq 1$, we say that $(m^+_z, m^-_z) \in \AA^*$ if $(m^+_z, m^-_z)\equiv(m^{\ell_<,+}_z, m^{\ell_<, -}_z)$  and this profile is taken on by (the $\ell_<$-piecewise constant extension to $\hat{\Lambda}$ of) some $(\splus_z, \sminus_z)$ such that $\sigma \in \AA$. Conversely, any open set $A \subset L^{\infty}\times L^{\infty}$ gives rise to an event $\AA \in \mathcal B^{\ell_<}_{0, \Lambda}$, by taking block averages of elements of $A$. Given $\sigma_{\Lambda^c}$, let
\[
\tilde{Z}^{\omega}(\AA|\sigma_{\Lambda^c})=\int_{\AA} \textd \nu(\sigma_{\Lambda}) e^{-\beta \tilde{\scrH}^{\omega}(\sigma_\Lambda| \sigma_{\Lambda^c})}.
\]
For a fixed $\ell_<$-piecewise constant magnetization profile $(m^+_z, m^-_z), z \in \hat \Lambda$ and $\xi_1> 0$, of particular interest are the events
\[
\OO(m^+_z, m^-_z; \xi_1):= \{\sigma_{\Lambda}: \|(\splus_z, \sminus_z) - (m^+_z, m^-_z)\|_{L^\infty(\Lambda)} < \xi_1\}.
\]

Let $D \subset \R^d$ be a finite union of $\ell_>$-measurable blocks, $(m^+_z, m^-_z)$ be a pair of vector functions in $L^{
\infty}(D)$ and let $\bar m_{D^c, z}$ be a boundary condition on $D^c$.  We assume all functions are bounded by $1$ in norm. Let
\begin{multline}
 F_{J_L, D , \epsilon}(m^+, m^-|m_{D^c} )= - \int_{D
 \times D}\frac{J_L(z, z')}{2}\bar{m}_z \cdot \bar{m}_{z'}  - \frac{\epsilon}{2} \int_D \hate_2\cdot (m^+_z-m^-_z)
 \\ - \int_{D
 \times D^c}\frac{J_L(z, z')}{2}\bar{m}_r\cdot m_{D^c, z'}- \frac{1}{2\beta} \int_{D}\textd z S(m^+_z) + S(m^-_z).
\end{multline}
Note that since $J_L(z, z')=0$ if $\|z-z'\|_2> L$, there is no loss in replacing $m_{D^c}$ by its restriction to
\[
\partial D_{\leq L}:= \{z \in D^c: \dist(z, D) \leq L\}.
\]

Given a spin configuration $\sigma \in D \cap \Z^d$, the functional
\[
F_{J_L, D , \epsilon}(\splus, \sminus|\sigma^{\ell_<}_{D^c} )
\]
is defined by extending the block averages from $D\cap \Z^d$ to $D$ in the obvious way.
\begin{theorem}
\label{T:Part}
Let $\Lambda$ be a bounded, $\ell_>$-measurable subset of $\Z^d$.  Then there exists a universal constant $c>0$ so that
\[
\log  \tilde Z(\AA|\sigma_{\Lambda^c}) \leq - \beta \inf_{(m^+, m^-)\in \AA^*} F_{J_L, \hat \Lambda}(m^+, m^-| \sigma^{\ell_<}_{\Lambda^c}) +C(\beta  L^{\lambda-1} + L^{-\frac 54(1-\lambda)} \log L )|\Lambda|
\]
Also, if $(m^+_{0,z}, m^-_{0, z})_{z \in \Lambda}$ is a $\ell_<$-piecewise constant two component magnetization profile, and $\xi_1 \geq L^{-\frac 54(1-\lambda)} \log L$
\begin{equation}
\log  \tilde Z(\OO(m^+_{0, z}, m^-_{0, z}; \xi_1)|\sigma_{\Lambda^c})
\geq - \beta F_{J, L, \hat \Lambda}(m^+_{0, z}, m^-_{0, z}| \sigma^{\ell_<}_{\Lambda^c}) - C(\beta L^{\lambda-1}+L^{-\frac 54(1-\lambda)} \log L )|\Lambda|
\end{equation}
\end{theorem}

\begin{proof}
Let us prove the upper bound.
To begin, we decompose $\{x\in \R^2 : \|x\|_2 \leq 1\}$ via a finite collection $\CC_L$ of open balls $B_{\delta}(m)$ which intersect at most $C$ number of times where $C$ is a universal constant independent of $\delta$. Recall that $h(\rho)= \textrm{argmin}_h (G(h\hate_1)-  \rho h)$ for $\rho\in (-1, 1)$. We require that this collection $\CC_L$ be fine enough so that
\[
\delta h(\|m\|_2) \vee \delta^{-2} \ell_<^{-2d} \leq c L^{-\frac 54(1-\lambda)} \log L
\]
for each balls' center $m$.  As $\ell_<= L^{1-\lambda}$, $\delta= L^{- \frac 54 (1-\lambda)}$ satisfies this condition by Lemma \ref{L:Ent}.  The number of balls required for $\CC_L$ is adequately bounded by $c L^2$.

Given $\AA \in \mathcal B^{\ell_<}_{0, \Lambda}$, we can find a cover $\CC_{\AA}$ of $\AA^*$ consisting of balls $C_j$ in $L^{\infty}\times L^\infty$ each with radius $\delta$. The centers of these balls are given by $\ell_<$-piecewise constant profiles $(m^+_z(j), m^-_z(j))$ taking values among the centers of the balls in $\CC_L$.
The number of balls needed in this cover is given by, at most, $|\CC_L|^{ \frac{2 |\Lambda|}{\ell_<^d}}$.  Since $|\CC_L| \leq c L^2$, this gives a total covering number of (at most)
\[
 (c L)^{ \frac{5 |\Lambda|}{\ell_<^d}}.
\]

By the above and using Proposition \ref{P:EAprox}
\begin{multline}
\label{S1}
\log \tilde Z(\AA|\sigma_{\Lambda^c})\leq  \max_{\{j: C_j  \in \CC_{\AA}\}} \left[ -\beta U^{(L)}(m^+(j), m^-(j)|\sigma^{\ell_<}_{\Lambda^c})
+ \log  \nu_{\Lambda}(\tilde C_j)\right]  \\
+ c (\beta  L^{\lambda-1}+  \ell_<^{-d}\log L) |\Lambda|.
\end{multline}
In the continuum energy, $\sigma^{\ell_<}_{\Lambda^c}$ is the natural $\ell_<$-piecewise constant extension to $\hat{\Lambda}^c$ and 
\[
\tilde C_j=\{ \sigma: (\splus_z, \sminus_z) \in C_j\}.
\]

Lemma \ref{L:Ent} and our choice of $\delta$ implies that
\begin{equation}
\label{S2}
\log  \nu_{\Lambda}(\tilde C_j) \leq c L^{-\frac 54(1- \lambda)} \log L  |\Lambda| + \\
\sum_{z \in \Lambda^+} S(m^+_z(j)) + \sum_{z \in \Lambda^-} S(m^-_z(j))
\end{equation}
where we used the notation $\Lambda^{\pm}= \cup_{r \in \Lambda} B^{\pm}_r$.
Combining \eqref{S1} and \eqref{S2}, we have
\begin{multline}
\label{S3}
\log \tilde Z(\AA|\sigma_{\Lambda^c})\leq - \beta \min_{j: C_j \in \CC_{\AA}}F_{J, L, \Lambda}(m^+(j), m^-(j)| \sigma^{\ell_<}_{\Lambda^c})  \\
+ c(\beta L^{\lambda-1}+  L^{-\frac 54(1- \lambda)} \log L ) |\Lambda|.
\end{multline}

Finally, the energy functional $U$ is uniformly Lipschitz while the entropy function $S(m)$ has the uniform modulus of continuity  $c \|x\|(1+ |\log \|x\||)$ on the domain $\{x: \|x\|_2 \leq 1\}$, so we have
\begin{multline}
\label{S4}
\log \tilde Z(\AA|\sigma_{\Lambda^c})\leq - \beta \inf_{(m^+, m^-)\in \AA} F_{J, L, \Lambda}(m^+, m^-| \sigma^{\ell_<}_{\Lambda^c})  \\
+ c(\beta L^{\lambda-1}+ L^{-\frac 54(1-\lambda)} \log L ) |\Lambda|.
\end{multline}

The lower bound is similar enough that its derivation omitted.
\end{proof}

\noindent
\textbf{Proof of Lemma \ref{L:Key}:Free Energy Functional Analysis}
To reduce the notation, in this subsection we do not distinguish between sets $\Lambda \subset \Z^d$ and their extensions $\hat \Lambda \subset \R^d$.

The following condition on $(\beta, \epsilon)$ will be important from now on:
\begin{equation}
\label{E:BE}
 4\beta e^{-c_1 \beta}  < \frac{c_2}{2} \epsilon^2.
\end{equation}
where $c_1, c_2$ are the constants from Proposition \ref{P:Contract}.
\begin{lemma}
\label{L:Stupid}
Let $\Lambda$ be a bounded, $\ell_>$ measurable subset of $\Z^d$.  There exist $\delta, \epsilon_0, \beta_0>0$ so that if $\beta> \beta_0$ and $0<\epsilon< \epsilon_0$ satisfy \eqref{E:BE}, $\xi< \delta$,  $p, \lambda, \kappa \in (0, \frac 13)$, and $L^{-\lambda}< \frac{1}{8}$. Then
\[
\log  \tilde{W}(\Gamma; \sigma_{\delta(\Gamma)^c}) \leq C_1\left\{ -  \beta c(\epsilon^2 \wedge \xi^{2} \wedge L^{-\lambda d}) \ell_{<}^d + C(\beta L^{\lambda-1} + \beta e^{-c \epsilon^2 L^{\lambda}} + L^{-\frac 54(1-\lambda)} \log L )\ell_>^d \right\}  N_{\Gamma}.
 \]
\end{lemma}
\begin{proof}
We proceed in two steps.
Using Theorem \ref{T:Part}, it is enough to obtain an appropriate upper bound on
\[
-\inf_{(m^+, m^-)\in E^*_0} F_{J_L, \Lambda}(m^+, m^-| \sigma^{\ell_<}_{\Lambda^c})
\]
where
\[
E_0= \left\{\sigma'_{\delta(\Gamma)}: \: \: \theta_z(\sigma'_{\delta(\Gamma)})= \theta_{\Gamma}(z) \text{ for } z \in \Sp(\Gamma) \text{ and } \Theta_z(\sigma'_{\delta(\Gamma)})= \pm 1\text{ for } z \in R^{\pm}(\Gamma)\right\}
\]
in terms of
\[
-\inf_{(m^+, m^-)\in B^*} F_{J_L, \Lambda}(m^+, m^-| \sigma^{\ell_<}_{\delta(\Gamma^c)})
\]
for some
\[
B \subset \{\Theta_z(\sigma'_{\delta(\Gamma)})= 1\text{ for } z \in \delta(\Gamma)\}.
\]

It is convenient to  make a reduction.
Let
\begin{align*}
\Gamma_{Strip}=&\left\{ z \in \delta(\Gamma): \min\left(\dist(z, \Sp(\Gamma)), \dist(z, \delta(\Gamma)^c)\right)
\geq \frac{\ell_>}{4}\right\},\\
\tilde{\Gamma}_{Strip}=&\left\{ z \in \delta(\Gamma): \min\left(\dist(z, \Sp(\Gamma)), \dist(z, \delta(\Gamma)^c)\right)
\geq \frac{\ell_>}{16}\right\}.
\end{align*}
All events of interest have the property that $\Theta$ is constant and nonzero over the connected components of $\tilde{\Gamma}_{Strip}$.
Let
\[
E_1^*=E^*_0 \cap\left \{(m^+_z, m^-_z)= \Theta_z \times (\mathbf {m^+}, \mathbf {m^-}) \: \: \forall z \in  \Gamma_{Strip}\right \}.
\]

Under these conditions we have the following:
\begin{proposition}
\label{P:Strip}
Let $\Gamma$ be a contour.  Let $ \sigma^{\ell_<}_{\delta(\Gamma)^c}$ be any boundary condition compatible with $\Gamma$.
Then there exists $(\psi^+, \psi^-) \in E_1^*$ so that
\[
F_{J_L, \Gamma}(\psi^+, \psi^-| \sigma^{\ell_<}_{\delta(\Gamma)^c}) \leq \inf_{(m^+, m^-)\in E_0^*} F_{J_L, \Lambda}(m^+, m^-| \sigma^{\ell_<}_{\delta(\Gamma)^c}) + Ce^{- c \epsilon^2 L^{\lambda}} |\Gamma_{strip}|.
\]
\end{proposition}
This is a version of a corresponding result in Chapter 4 of Presutti' s book \cite{Pres-Book}.  As such, we postpone the proof until the end of the paper.  However, from a technical perspective,  this bound provides the main extra cost to the range of interaction (besides the intrinsic limitations discussed in the introduction).

By the previous proposition, we may now work within $E_1$.
We can bring the mean field free energy functional $f(m^+, m^-)$ into the picture as follows:
Let
\[\delta(\Gamma)^0 = \delta(\Gamma) \cap \left\{ \dist(z, \delta(\Gamma)^c)> \frac{3\ell_>}{8}\right\}.
\]
Because we have assumed $L^{-\lambda} < \frac 18$, on $E_1$ the functional $F_{J_L, \delta(\Gamma), \epsilon}(m^+, m^-|  \sigma^{\ell_<}_{\delta(\Gamma)^c})$ decouples into two terms:
\begin{multline}
\label{E:Decouple}
F_{J_L, \delta(\Gamma), \epsilon}(m^+, m^-|  \sigma^{\ell_<}_{\delta(\Gamma)^c})=
 \int_{\delta(\Gamma)^0
 \times \delta(\Gamma)^0}\textd z \textd z' \frac{J_L(z, z')}{2}(\bar{m}_z- \bar{m}_{z'})^2  \\
 + \int_{\delta(\Gamma)^0} \textd z f(m^+_z, m^-_z)
  + \int_{\delta(\Gamma)^0
 \times \delta(\Gamma)^{0\: c}}\textd z \textd z' \frac{J_L(z, z')}{2}(\bar{m}_z - \Theta_{z'}(\Gamma)\bar{\mathbf {m}})^2 \\+ C_{\delta(\Gamma)^0}(m^+, m^-|\sigma^{\ell_<}_{\delta(\Gamma)^c})
 \end{multline}
where
\begin{multline}
\label{E:C}
C_{\delta(\Gamma)^0}(m^+, m^-|\sigma^{\ell_<}_{\delta(\Gamma)^c})=
\inf_{(m^+, m^-)} \phi(m^+,m^-)|\delta(\Gamma)| +\int_{\delta(\Gamma)\backslash \delta(\Gamma)^0} \textd z f(m^+_z, m^-_z) \\
+ \int_{\delta(\Gamma)\backslash \delta(\Gamma)^0 \times \delta(\Gamma)^0} \textd z \textd z' \frac{J_L(z, z')}{2} (\bar{m}_{z}- \Theta_{z'}(\Gamma)\bar{\mathbf m})^2 \\
+ \int_{\delta(\Gamma)\backslash \delta(\Gamma)^0 \times \delta(\Gamma)\backslash \delta(\Gamma)^0} \textd z \textd z' \frac{J_L(z, z')}{2} (\bar{m}_{z}- \bar{m}_{z'})^2 \\
+ \int_{\delta(\Gamma) \backslash \delta(\Gamma)^0 \times \delta(\Gamma)^c} \textd z \textd z' \frac{J_L(z, z')}{2}  (\bar{m}_{z} - \sigma^{\ell_<}_{\delta(\Gamma)^c, z'})^2 + E(\sigma^{\ell_<}_{\delta(\Gamma)^c, z'}) + O(L^{\lambda-1}|\delta(\Gamma)|).
\end{multline}
and the term $E(\sigma^{\ell_<}_{\delta(\Gamma)^c, z'})$ is quadratic in $\sigma^{\ell_<}_{\delta(\Gamma)^c, z'}$, independent of $(m^+, m^-)$ and invariant under simultaneous reflections of spins about the $Y$-axis in each connected component of $\delta(\Gamma)^c$.

Consider the result of reflecting
$m^{\ell_<, +}_z, m^{\ell_<, -}_z$ on $\delta(\Gamma)\backslash \delta(\Gamma)^0$ and $\sigma^{\ell_<}_{\delta(\Gamma)^c, z}$ on $\delta(\Gamma)^c$ about the $Y$-axis as necessary so that these block averages all have positive projection along the $X$ axis.  This transformation of magnetization profiles  leaves $C_{\delta(\Gamma)^0}$ invariant provided we redefine $\Theta_x \equiv 1$ in the second term of \eqref{E:C}.

Let us extend the transformed profile to all of $\delta(\Gamma)$ by setting it to $(\mathbf m^+, \mathbf m^-)$ on $\delta(\Gamma)^0$.
The term
\begin{multline}
\FF_{J_L, \delta(\Gamma)^0, \epsilon}(m^+, m^-)= \int_{\delta(\Gamma)^0
 \times \delta^0(\Gamma)}\textd z \textd z' \frac{J_L(z, z')}{2}(\bar{m}_z- \bar{m}_{z'})^2  \\
 + \int_{\delta(\Gamma)^0} \textd z f(m^+_z, m^-_z)
  + \int_{\delta(\Gamma)^0
 \times \delta(\Gamma)^{0 \: c}}\textd z \textd z' \frac{J_L(z, z')}{2}(\bar{m}_z - \Theta_{z'}\bar{\mathbf m})^2
\end{multline}
on the right hand side of \eqref{E:Decouple} thus
represents the free energy difference between the original configuration and this new configuration.

Since the new profile is in $\{\Theta_z(\sigma'_{\delta(\Gamma)})= 1\text{ for } z \in \delta(\Gamma)\}^*$, we will be finished once we show:
\begin{proposition}
With the hypotheses of Lemma \ref{L:Stupid},
there exists a universal constant $c>0$ so that we have
\[
\inf_{(m^+, m^-) \in E_1} \FF_{J_L, \delta(\Gamma)^0, \epsilon}(m^+, m^-) \geq  C(\epsilon^2 \wedge \xi^{2} \wedge L^{-\lambda d}) \ell_<^d N_{\Gamma}
\]
\end{proposition}
\begin{proof}
Let
\[
T^0_\Gamma=\{ B_r\subset \Sp(\Gamma) : \text{ $B_r$ is $\ell_>$-measurable and } \theta_z(\Gamma)|_{B_r}= 0\}
\]
\begin{multline}
T^{\pm}_\Gamma=\{ B_r\subset \Sp(\Gamma) : \text{ $B_r$ is $\ell_>$-measurable, }\theta_z(\Gamma)|_{B_r}= 1\\ \text{ and some $\ell_>$-measurable block neighbor has  $\theta_z(\Gamma)=-1$}\}.
\end{multline}

For $B_r \in T^0_\Gamma$ we have $\eta_z(\Gamma)= 0$ for some $\ell_<$-measurable block in $B_r$.  Since the profiles $(m^+, m^-)$ are $\ell_<$-constant, Theorem \ref{T:MFT} implies
\[
\int_{B_r} \textd z' f(m^+_{z'}, m^-_{z'}) \geq c\epsilon^2 \wedge \xi^2 \ell_<^d.
\]
Otherwise, $\eta_z(\Gamma)$ doesn't vanish over $B_r$ and so there are two adjacent $\ell_<$-measurable blocks in $B_r$ for which $\eta_z(\Gamma)$ takes on different nonzero values.  Call them $C_r, C_{r_1}$.  As $\beta> \beta_0$ and $\xi< \frac{1}{3}$, we can find a universal constant $c>0$ so that
\[
\int_{C_r \times C_{r_1} }\textd z \textd z' \frac{J_L(z, z')}{2}(\bar{m}_z-\bar{m}_{z'})^2 \geq c \ell_<^{2d}L^{-d}.
\]
Similarly for $B_r \in T^{\pm}$,
\[
\int_{B_r \times \R^d }\textd z \textd z' \frac{J_L(z, z')}{2}(\bar{m}_z- \bar{m}_{z'})^2 \geq c \ell_<^{2d}L^{-d}
\]
The estimate in the statement now follows from the fact that $N_\Gamma\leq c d (|T^0| + |T^{\pm}|)$.
\end{proof}
\noindent
This finishes the proof of Lemma \ref{L:Aprox}
\end{proof}

\noindent
\textbf{Proof of Lemma \ref{L:Key}: A Potential Flow Argument}
To reduce the notation, in this subsection we do not distinguish between sets $\Lambda \subset \Z^d$ and their extensions $\hat \Lambda \subset \R^d$.

\begin{proof}[ Proof of Proposition \ref{P:Strip}]
The proposition is the consequence of the following model computation.
Assume for definiteness that $\Gamma$ is a $+$ contour.  Fix $(m_0^+, m_0^-) \in E^*_{0}$ and recall that $\delta_{ext}(\Gamma)= \delta(\Gamma)\cap \Ext(\Gamma)$.  Let $\delta(\Gamma)_{Strip}=\tilde{\Gamma}_{Strip} \cap  \delta_{ext}(\Gamma)$.  Below we restrict attention to $\delta(\Gamma)_{Strip}$, but the same sort of analysis may be carried out over the components of $\delta_{in}(\Gamma)$.

We may consider
\[
F_{J_L, \delta(\Gamma_{Strip})}(m_1^+, m_1^-| \bar m_0)
\]
on $L^{\infty}(\delta(\Gamma)_{Strip}) \times L^{\infty}(\delta(\Gamma)_{Strip}) $ subject to the boundary condition induced by restricting $(m_0^+, m_0^-)$ to $\delta(\Gamma) \cap \delta(\Gamma)_{Strip}^c$ and subject to the constraint that
the profile
\[
(m^+, m^-):= \begin{cases}
(m_1^+, m_1^-) \text{ on $\delta(\Gamma)_{Strip}$,}\\
(m_0^+, m_0^-) \text{ otherwise,}
\end{cases}
\]
lies in $E^*_0$.  By definition, this is equivalent to requiring that $\Theta(m_1^+, m_1^-) \equiv 1$ on $\delta(\Gamma)_{Strip}$.

For each $m_2 \in \{x \in \R^2: \|x\|_2 < 1\}$, let us define $M^{-1}(m_2)=- \nabla S(m_2)$.  Because $S(m_2)$ is strictly concave in this region, the inverse $M(h)$ is well defined.  It takes the explicit form
\[
M(h)= \frac{\int  \sigma e^{\sigma \cdot h} \textrm{d}\nu(\sigma) }{\int  e^{\sigma \cdot h} \textrm{d}\nu(\sigma)}
\]
for all $h \in \R^2$.

We consider the differential equation:
\begin{equation}
\label{E:Co1}
\partial_t m^{\pm} = -\left[m^{
\pm} - M(\beta(J_L\star \bar{m}\pm \epsilon \hate_2))\right]
\end{equation}
on $\delta(\Gamma)_{Strip}$
subject to the boundary conditions
\begin{equation}
\label{E:Co2}
(m^+_z(t), m^-_z(t))\equiv\: (m_0^+, m_0^-) \text{ for $(t, z) \in \{ t= 0\} \times \delta(\Gamma)_{Strip} \cup \{ t\geq 0\} \times \delta(\Gamma)_{Strip}^c$}.
\end{equation}
Here, $J_L\star \bar{m}$ is the convolution of $J_L$ with $\bar m$:
\[
J_L\star \bar{m}_z:= \int_{\R^d} J_{L}(z, r) \bar m_r \textd r.
\]

The stationary solutions to this differential equation satisfy
\begin{equation}
\label{E:Station}
m^{\pm}_z - M(\beta(J_L\star \bar{m}_z\pm \epsilon\: \hate_2)).
\end{equation}
Note that $(\mathbf m^+, \mathbf m^+)$ solves the vector equation
\[
\mathbf{m}^{\pm} = M(\beta(\bar{\mathbf{m}}\pm \epsilon \hate_2))
\]
so we may think of \eqref{E:Station} as a kind of generalized mean field equation.

\begin{lemma}
\label{L:contract}
There exist $\delta, \epsilon_0, \beta_0>0$ so that if $\beta> \beta_0$ and $\epsilon< \epsilon_0$ satisfy \eqref{E:BE}, $\xi \in (0, \delta)$ and $(m_0^+, m_0^-) \in E^*_0$.  Then we have existence and uniqueness of \eqref{E:Co1}, \eqref{E:Co2} for all time $t\in [0, \infty)$ in $L^\infty(\delta(\Gamma)_{Strip}) \times L^\infty(\delta(\Gamma)_{Strip})$.  Further, for each initial condition $(m_0^+, m_0^-) \in E_0$, there is a unique stationary solution to \eqref{E:Station} with boundary condition $(m_0^+, m_0^-)|_{\delta(\Gamma)_{Strip}^c}$.  Its $\ell_<$-measurable block average profile lies in $E_0^*$.
\end{lemma}
\begin{proof}
For a start let us assume that $\delta, \epsilon_0, \beta_0>0$ as in Proposition \ref{P:Contract}.   We will adjust $\beta_0$ in a moment.
The function $m \mapsto M(\beta(m\pm \epsilon \hate_2))$ is $C$-Lipschitz within $\{x: \|x- \bar{\mathbf m}\|_2 < \xi\}$ since 
\[
\nabla_h M(\beta h= \beta \nabla M(\beta h) = \beta \left(\langle \sigma \hat \cdot \sigma \rangle - \langle \sigma \rangle \hat \cdot \langle \sigma \rangle\right)
\]
where $ \hat \cdot $ denotes the outer product of two vectors and the expectation is taken under $\nu$ tilted by $e^{\beta \sigma \cdot h}$.  In fact, increasing $\beta_0$ as necessary
\[
\|\beta \nabla M(\beta J\star \bar{m}\pm \epsilon \hate_2)\|<1
\]
on $E$ since in this case $\beta J\star \bar{m}\pm \epsilon \hate_2$ are uniformly bounded away from $0$.  
So if
\[
K^{\pm}(m^+, m^-)= m^{\pm} - M(\beta(J\star \bar{m}\pm \epsilon \hate_2))
\]
then $K^\pm$ are (at worst) $2$-Lipschitz (in $L^{\infty}$).
We will use this at the end of the proof as well.

Using the Lipschitz continuity, we have global existence and uniqueness on $L^{\infty} \times L^{\infty}$ for any initial condition in $E^*_0$.
Let
\[
E=\{(m^+, m^-)\in L^{\infty} \times L^{\infty}:  \|\bar m - \bar{\mathbf m}\|_{L^{\infty}}< \xi,\: \|m^{\pm}\|_{L^{\infty}}< 1, m^{\pm}\cdot \hate_1 > 0 \}.
\]
Note that $E^*_0 \subset E$.
We first claim that the dynamics preserves $E$:

By continuity in $t$, given $T_0$ sufficiently small and positive, we may assume that on the interval $[0, T_0]$ the solution is in $E$.  Note that the constant profile $(a^+, a^-) \equiv (\mathbf m^+, \mathbf m^-)$ is in $E^*_0$.
Since this profile is stationary we have
\begin{multline}
\frac{1}{2} \partial_t \|\bar a(t)- \bar m_2(t)\|_2^2= - \|\bar a- \bar m_2\|_2^2
+ \left[\bar a- \bar m_2\right] \cdot \\
\left[ \frac{1}{2}(M(\beta(J_L \star \bar{a}+ \epsilon \hate_2)+ M(\beta(J_L \star \bar{a}- \epsilon \hate_2)) - \frac{1}{2}( M(\beta(J_L \star \bar{m}_2+ \epsilon \hate_2))+ M(\beta(J_L \star \bar{m}_2- \epsilon \hate_2))) \right].
\end{multline}
By assumption on $\beta$ and Proposition \ref{P:Contract}
\begin{multline}
\| \frac{1}{2}(M(\beta(J_L \star \bar{a}+ \epsilon \hate_2)+ M(\beta(J_L \star \bar{a}- \epsilon \hate_2)) - \frac{1}{2}( M(\beta(J_L \star \bar{m}_2+ \epsilon \hate_2))+ M(\beta(J_L \star \bar{m}_2- \epsilon \hate_2))) \|_2
\\
\leq (1- c \epsilon^2)\| J_L\star \bar a - J_L \star \bar{m}_2\|_2
\end{multline}
for some $c>0$.
Since $J \star$ is a contraction in $L^{\infty}$,
\[
\frac{1}{2} \partial_t \|\bar a- \bar m_2\|_2^2= - \|\bar a- \bar m_2\|_2^2
+ (1- c \epsilon^2) \xi \|\bar a- \bar m_2\|_2
\]
for all $t \leq T_0$.  By a continuity argument, if \textit{a priori} $\|\bar a(0)- \bar m_2(0)\|_2< \xi$ it remains so for all times $t\in [0, \infty)$.

Consider next the boundedness of the individual components $m_t^{\pm}$.  The condition $m^{\pm} _t\cdot \hate_1> 0$ is evidently preserved by considering the flow velocity along with the fact that $\|\bar {m} - \bar a\|_2 < \xi \: \: \: \forall t \geq 0$.  If the initial condition satisfies $\|m^{\pm}_0\|_{L^{\infty}}< 1$,  then since $\|M(\beta(J_L \star \bar{m}_2\pm \epsilon \hate_2)))\|_2$ is uniformly bounded away from  $1$ (on $E$), it is easy to see that the dynamics is a coordinate wise contraction when either component lies in $\rho< \|m^{\pm}_0\|< 1$ for some $\rho$ sufficiently close to one.
We conclude via another continuity argument that the dynamics preserves $E$.

Under the flow, the free energy functional $F_{J_L, \delta(\Gamma)_{Strip}}(m^+(t), m^-(t)| \bar m_0)$  subject to initial conditions lying in $E$, evolves as:
\begin{multline}
\partial_t F_{J_L, \delta(\Gamma)_{Strip}} =\\
\frac{1}{2} \int \textd r \left[ J \star m +  \epsilon \hate_2 + \nabla S( m^{+})\right]  \cdot \left[m^{+} - M(\beta(J\star \bar{m}+  \epsilon \hate_2))\right]\\
+\frac{1}{2} \int \textd r \left[ J \star m  - \epsilon \hate_2 + \nabla S( m^{-})\right]  \cdot \left[m^{-} - M(\beta(J\star \bar{m}- \epsilon \hate_2))\right]
\end{multline}
Now  $M$ and $-\nabla S$ are inverses of one another when the latter is restricted to $\|m\|_2<1$, so each integrand is of the form
\[
-(M(m)- M (h))\cdot (m-h).
\]
Since $\nabla M \geq  0$ as an operator,  we conclude that $\partial_t F_{J_L, \delta(\Gamma)_{Strip}}$ is decreasing.  It is strictly decreasing unless  $(m^+, m^-)$ satisfies \eqref{E:Station} because
\[
v \cdot \nabla M  v= \textrm{Var}_h(\sigma \cdot v) \neq 0
\]
for any unit vector $v$ and any $h \in \R^2$.  Here $\textrm{Var}_h$ denote the variance of a random variable under the measure defined by tilting $\nu$ by $e^{\sigma \cdot h}$.  It is easy to see that the functional is bounded below on $E$, so we conclude that any limit point of the flow started from an element of $E$ must be a stationary point.

From \eqref{E:Station}, we have
\[
m^{\pm}= M(\beta(J_L\star \bar{m}\pm \epsilon \hate_2))
\]
so that the two components of any stationary point MUST be determined by the average vector $\bar{m}$.  In particular using the representation \eqref{E:M*}
\[
\|m^{\pm}_z-\mathbf{m}^{\pm}\|_2= \|M(\beta(J_L\star \bar{m}_z\pm \epsilon \hate_2))- M(\beta(J_L\star \bar{\mathbf m}\pm \epsilon \hate_2))\|_2 \leq \| \bar{m}_z- \bar{\mathbf m}\|_{L^\infty}
\]
where we have used that $M(\beta v\pm \epsilon \hate_2)$ is $1$-Lipschitz in $v$ and that $J_L\star$ is a contraction in $L^{\infty}$.
We have proved that any stationary solution has block averages in $E^*_0$.
\end{proof}

The next order of business is to show that if $m^{\pm}_r$ is a stationary point with boundary condition $m_0^{\pm}$ then $\bar{m}_r$ is very close to $\bar{\mathbf m}$ if $r$ is deep inside $\delta(\Gamma)_{Strip}$.
\begin{lemma}
With the same hypotheses as in the previous proposition, let $(m^{+}_r, m^{-}_r)$ be a stationary solution to \eqref{E:Station} on $\delta(\Gamma)_{Strip}$ subject to the boundary condition coinciding with $(m^+_0, m^-_0)\in E_0$ on $\delta(\Gamma)_{Strip}^c$.  Then there exist universal constants $C, c>0$ so that
\[
\|\bar m_r- \bar{\mathbf m}_r\|_2 < C e^{- c \epsilon^2 L^{\lambda}}
\]
whenever $d(r, \delta(\Gamma)_{Strip}^c)>\frac {\ell_>} 8 $
\end{lemma}
\begin{proof}
If $m\in E$ is stationary then (using notation from Section \ref{S:MFT}),
\[
\bar{m}- \bar{\mathbf m}= M^*_{\beta}( \bar{m} ) -  M^{\star}_\beta(\bar{\mathbf m}) .
\]
Hence by Proposition \ref{P:Contract}, if on some region $D$ $\|\bar{m}- \bar{\mathbf m}\|_{L^{\infty}(D)}< \xi$ then on the interior $D_L^o= \{x \in D: \dist(x, D^c)> L\}$
\[
\|\bar m- \bar{\mathbf m}\|_\infty < (1- c \epsilon^2) \xi.
\]
Iterating, if  $\dist(r, D^c)> k L$ then
\[
\|\bar m_r- \bar{\mathbf m}\|_2 < (1- c \epsilon^2)^k \xi.
\]
Hence we conclude that any stationary solution has
\[
\|\bar m_r- \bar{\mathbf m}\|_2 < C e^{- c \epsilon^2 L^{\lambda}} \xi
\]
if $\dist(r, \delta(\Gamma)_{Strip}^c)> \frac{\ell_>}{8}$.
Since the operator norm of $\nabla M^*_{\beta}$ is bounded by $1$ the same holds for $m^\pm$ using \eqref{E:Station}.
\end{proof}

Finally, if $(m^{+}_z, m^{-}_z)$ is a stationary solution to \eqref{E:Station} on $\delta(\Gamma)_{Strip}$ subject to the boundary condition coinciding with $(m^+_0, m^-_0)\in E_0$ on $\delta(\Gamma)_{Strip}^c$, then its block averages $(m^{\ell_<, +}_z, m^{\ell_<, -}_z)$ lie in $E^*_0$ and satisfy
\[
\|m^{\ell_<, \pm}_z- \mathbf m^{\pm}\|_2 < C e^{- c \epsilon^2 L^{\lambda}}
\]
if $\dist(r, \delta(\Gamma)_{Strip}^c)> \frac{\ell_>}{8}$.
From the proof of Lemma \ref{L:contract}
\[
F_{J_L, \delta(\Gamma)_{Strip}}(m^+, m^-| \bar m_0)= \inf_{(m_1^+, m_1^-)\in E}  F_{J_L, \delta(\Gamma)_{Strip}}(m^+, m^-| \bar m_0)
\]

Note that
\[
\|J\star\bar m_z- J\star \bar m^{\ell_<}_z\|_2 \leq c L^{\lambda-1}
\]
and that
\[
\int_{|z-r|\leq \ell_<}\textd z f(m_z^{+}, m_z^-) \geq \ell_<^d  f(m^{\ell_<, +}_z, m^{\ell_<, -}_z)
\]
on $E_0^*$ since $f$ is convex on $\{\|(m^{+}, m^-)\pm ({\mathbf m^+, \mathbf m^-})\| < \xi\}$.

Thus we have
\[
F_{J_L, \delta(\Gamma)_{Strip}}(m^{\ell_<, +}, m^{\ell_<, -}| \bar m_0)- F_{J_L, \delta(\Gamma)_{Strip}}(m^+, m^-| \bar m_0) \leq c L^{\lambda-1} |\delta(\Gamma)_{Strip}|.
\]

To summarize, we have shown the existence of an $\ell_<$-constant magnetization profile $(m^{\ell_<, +}_z, m^{\ell_<, -}_z)$ so that
\begin{multline}
F_{J_L, \delta(\Gamma)_{Strip}}(m^{\ell_<, +}, m^{\ell_<, -}| \bar m_0) \\
\leq  \inf_{(m_1^+, m_1^-)\in E}  F_{J_L, \delta(\Gamma)_{Strip}}(m^+, m^-| \bar m_0)+ c L^{\lambda-1} \log L |\delta(\Gamma)_{Strip}|
\end{multline}
and
\[
\|m^{\ell_<, \pm}_z- \mathbf m^{\pm}\|_2 < C e^{- c \epsilon^2 L^{\lambda}}
\]
if $\dist(z, \delta(\Gamma)_{Strip}^c)> \frac{\ell_>}{8}$.
To complete the argument, we simply modify $m^{\ell_<, \pm}_z$ on $\dist(z, \delta(\Gamma)_{Strip}^c)> \frac{\ell_>}{8}$ to be \textit{equal} to $\mathbf m^{\pm}$.
Call the new profile $(m^{\ell_<, +}_{a, z}, m^{\ell_<, -}_{a,z})$.
Then
\begin{multline}
F_{J_L, \delta(\Gamma)_{Strip}}(m_a^{\ell_<, +}, m_a^{\ell_<, -}| \bar m_0) \\
\leq  \inf_{(m_1^+, m_1^-)\in E}  F_{J_L, \delta(\Gamma)_{Strip}}(m_1^+, m_1^-| \bar m_0)+ c L^{\lambda-1} |\delta(\Gamma)_{Strip}| + C e^{- c \epsilon^2 L^{\lambda}}  |\delta(\Gamma)_{Strip}|
\end{multline}
since the on-site term in $F_{J_L, \delta(\Gamma)_{Strip}}$ can only be decreased by this modification.
This concludes the proof of Proposition \ref{P:Strip}.
\end{proof}

\noindent
\textbf{Acknowledgements:}
The author would like to acknowledge interesting and helpful discussions with  D.~Ioffe and S.~Shlosman on the behavior of this model with nearest neighbor coupling and with D.~Ioffe on the generic behavior of  models of this type in the mean field setting.

\end{document}